\documentclass[11pt,letterpaper]{llncs}

\usepackage[table]{xcolor}
\usepackage{float}
\usepackage{fullpage}
\usepackage{graphicx}
\usepackage{xcolor}
\usepackage{amsmath}
\usepackage{amssymb}
\usepackage[sort,nocompress]{cite}
\usepackage[ruled,vlined,lined,boxed,commentsnumbered]{algorithm2e}
\usepackage{todonotes}
\usepackage{verbatim}
\usepackage{subcaption}
\usepackage{array}

\usepackage{pgf}
\usepackage{tikz}
\usetikzlibrary{arrows,automata,positioning}

\begin{document}

\title{On Indexing and Compressing Finite Automata}

\author{Nicola Cotumaccio\inst{1} \and Nicola Prezza\inst{2}\thanks{Corresponding author.}}
\institute{Gran Sasso Science Institute, L'Aquila, Italy, \email{nicola.cotumaccio@gssi.it}
\and
Luiss Guido Carli University, Rome, Italy, \email{nprezza@luiss.it}
}

\maketitle

\begin{abstract}
An index for a finite automaton is a powerful data structure that supports locating paths labeled with a query pattern, thus solving pattern matching on the underlying regular language. The problem is hard in the general case: a recent conditional lower bound suggests that, in the worst case, deciding \emph{at query time} whether a pattern of length $m$ belongs to the substring closure of the language accepted by $\mathcal A$ requires $\Omega(m\cdot |\mathcal A|)$ time. On the other hand, Gagie et al. [TCS 2017] introduced a subclass of automata that allow an optimal $\tilde O(m)$-time solution based on prefix-sorting the states in a total order. 
In this paper, we  solve the long-standing problem of indexing arbitrary finite automata, matching the above bounds. Our solution consists in finding a \emph{partial} co-lexicographic order of the states and proving, as in the total order case, that states reached by a given string form \emph{one interval} on the partial order, thus enabling indexing. We provide a lower bound stating that such an interval requires $O(p)$ words to be represented, $p$ being the order's width (i.e. the size of its largest antichain). Indeed, we show that $p$ determines the complexity of several fundamental problems on finite automata:
\begin{enumerate}
\item[(i)]  Letting $\sigma$ be the alphabet size, we provide an encoding for NFAs using $\lceil\log \sigma\rceil + 2\lceil\log p\rceil + 2$ bits per transition and a smaller encoding for DFAs using $\lceil\log \sigma\rceil + \lceil\log p\rceil + 2$  bits per transition. This is achieved by generalizing the Burrows-Wheeler transform to arbitrary automata. 
\item[(ii)] We show that indexed pattern matching can be solved in $\tilde O(m\cdot p^2)$ query time on NFAs.
\item[(iii)] We provide a polynomial-time algorithm to index DFAs, while matching the optimal value for $ p $. On the other hand, we prove that the problem is NP-hard on NFAs.
\item[(iv)] We show that, in the worst case, the classic powerset construction algorithm for NFA determinization generates an equivalent DFA of size $2^p(n-p+1)-1$, where $n$ is the number of NFA's states. 
\end{enumerate}
Contribution (i) provides a new compression paradigm for labeled graphs. Contributions (ii)-(iii) solve the regular language indexing problem, notably with a polynomial-time solution for DFAs. Contribution (iv) implies a new FPT analysis for the complexity of classic algorithms on automata, including membership and equivalence (the latter being PSPACE-complete when input automata are NFAs).
\end{abstract}

\newpage

\clearpage
\setcounter{page}{1}

\section{Introduction}

Sorting is arguably one of the most basic and, at the same time, powerful techniques when it comes to searching and compressing data: for instance, a sorted list of integers supports fast membership queries and is more compressible than any of its permutations. One of the major algorithmic breakthroughs of the last two decades is that this simple observation holds also in the string processing domain: in the year 2000, Ferragina and Manzini \cite{ferragina2000opportunistic} and Grossi and Vitter \cite{grossi2000} showed independently that the list of co-lexicographically sorted prefixes of a string\footnote{More precisely, their ending positions in the string. The original work moreover considered the lexicographically-sorted suffixes; in this work we adopt this symmetric point of view which is easier to generalize to finite automata.} can be used to support fast pattern matching queries (that is, counting and locating occurrences of a pattern in the string) while \emph{simultaneously} being compressed to the string's entropy. 
Subsequent works demonstrated that prefix sorting could be extended to nonlinear structures as well. In the year 2005, Ferragina et al. \cite{Ferragina2005} showed that the nodes of a labeled tree could be arranged in the co-lexicographic order of the paths connecting them to the root, and that the resulting sorted list could be used to index and compress the tree. Two years later, Mantaci et al. \cite{mantaci2007extension} applied the same principle to sets of strings. This fascinating journey countinued with the works of Bowe et al. \cite{BOSS} and Sir\'en et al. \cite{GCSA}, who extended the class of indexable graphs to (a generalization of) de Bruijn graphs. 
More recently, Gagie et al. \cite{GAGIE201767} showed a framework capturing all the above techniques in a unified theory: Wheeler graphs. 
The idea underlying this universal framework is to extend the prefix-sorting axioms from strings to graphs: node pairs are sorted by their incoming labels and, if the labels are equal, by their predecessors (i.e. the order propagates forward when following equally-labeled pairs of edges). 
While this class of graphs admits essentially-optimal indexing and compression strategies, it has a fundamental limitation: very few labeled graphs admit a total order of their nodes satisfying the above prefix-sorting axioms. As a matter of fact, languages recognized by finite automata whose state transition is a Wheeler graph are very simple: they are star-free and closed only by intersection, and every Wheeler NFA admits an equivalent Wheeler DFA of linear size \cite{alanko20regular,alanko2020wheeler} (compare this result with the exponential blow-up of the general case).

In this paper, we finally generalize prefix sorting to arbitrary finite automata. Our technique allows us to index and compress any finite automaton by exploiting its inherent \emph{sortability}, matches recent lower- and upper- bounds \cite{equi2019,equi2020conditional,GAGIE201767}, and has unexpected deep consequences in automata theory. 
In order to fully appreciate the contribution of our paper, it is instructive to consider the well-studied  problems of compressing permutations and adaptive sorting \cite{BARBAY2013109,barbayLRM,yehuda1998partitioning}. One of the most powerful techniques to achieve the former goal is to exploit the \emph{sortedness} of the permutation. While not all permutations of $[1,n]$ are totally sorted (in fact, only one permutation has this property: $1,2,\dots, n$), one can decompose an arbitrary permutation into monotone subsequences and compress them independently. The same idea can be applied to sorting integers: by decomposing an integer sequence into monotone subsequences, one can exploit the sortedness of the input in order to obtain a faster \emph{adaptive} sorting algorithm. The new  paradigm presented in this paper can be seen as an extension of the above ideas to finite automata: in our case, we sort the automaton's states according to the co-lexicographic order of the corresponding language's prefixes.
The key observation that we provide, needed for this idea to work properly, is that, while not all finite automata admit a \emph{total} co-lexicographic order of their states, a \emph{partial} order always exists. We show that the similarity between this order and a total one (more precisely, the order's width) is the key parameter that dictates how efficiently many fundamental problems on finite automata can be solved. 

\subsection{Our Results}

Given a finite automaton, we define an order $\leq$ of its states reflecting the co-lexicographic order of the strings read from the source. While a restricted class of automata admits a total order of this kind \cite{GAGIE201767,alanko2020wheeler}, in general $\leq$ is partial. Given such an order, we show that the states reached by a path labeled with a given string form \emph{one interval} (that is, a convex set) on the partial order. This is exactly what enables indexing and compression when the order is total \cite{GAGIE201767}. 
Let $p$ be the order's width, i.e. the size of its largest antichain. We start with a lower bound stating that intervals on the partial order require $\Theta(p)$ words each to be represented. 
While this implies that handling intervals is a bottleneck, we observe that a recent graph-indexability lower bound of Equi et al. \cite{equi2020conditional,equi2020graphs} indicates that a slow-down of this kind is inevitable.
We call an automaton \emph{$p$-sortable} if there exists a co-lexicographic partial order of its states of width at most $p$. We show that the parameter $p$ is an important complexity measure for several fundamental problems on NFAs. 

To begin with, we show that $p$ is a valid compressibility measure. 
We generalize the Burrows-Wheeler transform (BWT) \cite{burrows1994block} to NFAs, and prove that this transformation can be stored using $\lceil \log \sigma \rceil + 2\lceil\log p\rceil + 2$ bits per transition, where $\sigma$ is the alphabet's size. Surprisingly, we show that DFAs admit a smaller encoding: $\lceil \log \sigma \rceil + \lceil\log p\rceil + 2$ bits per transition.
We further show that this transform supports pattern matching as well: in Theorem \ref{th:FM-index}, we generalize the successful FM-index \cite{ferragina2000opportunistic} to NFAs. This solves the major open problem of indexing arbitrary finite automata. Our index uses $\lceil \log \sigma \rceil + \lceil\log p\rceil + 2$ bits per transition and supports counting the states reached by a query pattern $P\in \Sigma^m$ in time $O(m\cdot p^2 \cdot \log(p\cdot \sigma))$. 
Before our paper, only Wheeler automata \cite{GAGIE201767} (i.e. the case $p=1$) admitted an efficient solution for this problem.
In the worst case, our query time matches (up to a logarithmic factor) the lower bound $\Omega(m\cdot |\mathcal A|)$ of Equi et al. \cite{equi2020graphs,equi2020conditional}.

We proceed by showing that parameter $p$ also determines the "amount of nondeterminism" in NFAs: in Theorem \ref{th:powerset} we prove that the classic powerset algorithm for NFA determinization, when run on a $p$-sortable NFA with $n$ states, generates an equivalent DFA with at most $2^p(n-p+1)-1$ states.
This result has surprising implications. For instance, it shows that the PSPACE-complete NFA equivalence problem \cite{stockmeyer1973word} is fixed-parameter tractable with respect to $p$.
Similarly, our bound implies a better analysis of the powerset-based membership algorithm for NFAs.

Motivated by the above applications of our new technique, we conclude the paper by focusing on the problem of determining, given a pair $(\mathcal A,p)$, whether the NFA $\mathcal A$ is $p$-sortable. 
A recent result of Gibney and Thankachan \cite{gibney19} implies that the problem is NP-hard in the general case. On the other hand, in Corollary \ref{cor:poly-DFA} we provide a polynomial-time algorithm for DFAs. Notably, this result enables indexing arbitrary regular languages in polynomial time with the optimal value for $p$, provided that the input language is expressed as a DFA. 

Our approach opens up promising lines of research both in regular language theory and in compressed indexing. Our parameterization defines a complete hierarchy of sub-regular languages, i.e. those accepted by $p$-sortable NFAs. Moreover, it seems natural to expect that other fundamental problems such as NFA minimization \cite{Malcher2004} and regular expression matching \cite{backurs2016regular} admit efficient algorithms for small values of $p$ (similarly to what we proved for compression, indexing, equivalence, and membership). 
Further lines of research include determining the complexity of finding a co-lexicographic order of provably small width, 
recognizing $p$-sortable NFAs for a fixed $p$ (for $p=1$ the problem is known to be NP-complete \cite{gibney19}), and refining the lower bounds of Equi et al. \cite{equi2019,equi2020conditional,equi2020graphs} for the graph indexing problem as a function of $p$.

The first ten pages of this manuscript contain a concise description of all our contributions and can be followed by non-specialists. 
The detailed proofs of all claims can be found in the appendix.

\section{Notation}

A nondeterministic finite automaton (NFA) is a 5-tuple $ (Q, E, \Sigma, s, F) $ where $ Q $ is the set of states, $ E \subseteq Q \times Q \times \Sigma $ is
the automaton's transition function, $ \Sigma $ is the alphabet, $ s \in Q $ is the initial state and $ F \subseteq Q $ is the set of final states. 
We assume the alphabet to be effective: each character labels at least one edge and, in particular, $\sigma = |\Sigma| \leq |E|$.
A deterministic finite automaton (DFA) is an NFA such that each state has at most one outgoing edge labeled with a given character.
We make the same assumptions of Alanko et al. \cite{alanko20regular}: (i) We assume that our input NFAs are \emph{input-consistent}, that is, all edges reaching the same state have the same label. This is required for indexing and is not restrictive since input-consistency can be forced by replacing each state with $ |\Sigma | $ copies of itself without changing the accepted language.
(ii) We assume that all states are reachable from the initial state.
(iii) We assume that the (unique) initial state has no incoming edges.
(iv) We assume that every state is either final or it allows to reach a final state. (v) We do not require each state to have an outgoing edge for all possible labels.
It is not hard to see that these assumptions are not restrictive, since any automaton can modified to meet these requirements while preserving the accepted language.
We assume that on $ \Sigma $ there is a fixed total order $ \le $, and that strings in $ \Sigma^* $ are co-lexicographically ordered by $ \le $. We write $ a \le b $ when the string $ a $ is co-lexicographically smaller than or equal to the string $ b $, and we write $ a < b $ when $ a \le b $ and $ a \not = b $.
To simplify our notation, we denote by $\lambda(u)$ the (uniquely determined) label of all incoming edges of node $u$. For the initial state $s$, we write $\lambda(v) = \# \notin \Sigma$ and we assume $\# < c$ for all $c\in\Sigma$. To make notation more compact, we will sometimes write $ (u, v) $ for $ (u, v, a) $, because it must be $ a = \lambda (v) $.

A partial order $ \le $ on a set $ V $ is a reflexive, antisymmetric and transitive relation on $ V $. We write $ u < v $ when $ u \le v $ and $ u \not = v $. In particular, at most one between $ u < v $ and $ v < u $ can be true. 
Elements $u$ and $v$ are \emph{$ \le $-comparable} if $u \leq v$ or $v\leq u$ holds. We write $u\ \|\ v $ when $ u $ and $ v $ are not $ \le $-comparable (note that $\|$ is a symmetric relation). On a partial order $\leq$, for every $ u, v \in V $ exactly one of the following is true: (i)  $u=v$, (ii) $u<v$, (iii) $v<u$, or (iv)  $u\ \|\ v$.

If $ V' \subseteq V $, then we say that $ U \subseteq V' $ is a \emph{$\le_{V'} $-interval} if for every $ u, v, z \in V' $ such that $ u < v < z $ and $ u, z \in U $ we have $ v \in U $. In particular, we say that $ U \subseteq V $ is a \emph{$\le $-interval} if it is a $ \le_V $-interval. 

A subset $ Z \subseteq V $ is a \emph{$ \le $-chain} if $ (Z, \le) $ is a total order. A partition $ \{V_i \}_{i = 1}^m $ of $ V $ is a \emph{$ \le $-chain decomposition} if $ V_i $ is a $ \le $-chain, for every $ i = 1, \dots, m $. The \emph{$ \le $-width} of $ V $ (equivalently, the width of $(V,\leq)$ or simply the width of $\leq$ when $V$ is clear from the context) is the size of its largest antichain, i.e. the largest subset $A = \{u_1, \dots, u_p\} \subseteq V$ such that $u_i\ \|\ u_j$ for all $1\leq i< j \leq p$. Dilworth's theorem \cite{dilworth2009decomposition} states that the width of $(V,\leq)$ coincides with the cardinality of a smallest $ \le $-chain decomposition of $ V $.

Our results hold in the word RAM model with word size $\Theta(\log n)$ bits. Logarithms are base 2.

\section{Extending Prefix Sorting to Arbitrary Finite 
Automata}

We start by extending the notion of co-lexicographic order to the states of an arbitrary finite automaton. Crucially, note that in Axiom 2 the implication follows the edges backwards (instead of forward as done in \cite{GAGIE201767}). 
The reason for this will be made clear after the definition. 

\begin{definition}\label{def:co-lex order}
Let $\mathcal A = (Q, E, \Sigma, s, F) $ be an NFA. A \emph{co-lexicographic order} of $ \mathcal A $ is a partial order $ \le $ on $ Q $ that satisfies the following two axioms:
\begin{enumerate}
    \item (Axiom 1) For every $ u, v \in Q $, if $\lambda(u) < \lambda(v)$, then $ u < v $ (in particular, states with no incoming edges come before all remaining states);
    \item (Axiom 2) For all edges $ (u', u), (v', v) \in E $, if $ \lambda (u) = \lambda (v) $ and $ u < v $, then $ u' \leq v' $.
\end{enumerate}
\end{definition}

It is immediate to observe that a co-lexicographic order $ \le $ is a Wheeler order (as defined in \cite{GAGIE201767}) if and only $ \le $ is total. We remind the reader that an automaton is said to be Wheeler if and only if it admits a Wheeler order \cite{GAGIE201767,alanko2020wheeler}. 
The intuition behind Axiom 2 is that ensuring backward compatibility guarantees that the order is automatically not defined if predecessors cannot be unambiguously compared, 
as observed in the following remark.

\begin{remark}\label{rem1}
Let $\mathcal A = (Q, E, \Sigma, s, F) $ be an NFA and let $ \leq $ be a co-lexicographic order of $ \mathcal A $. Let $ u, v \in Q $ such that $ u \not = v $ and $ \lambda (u) = \lambda (v) $. Then, $ u~\|~v $ if at least one of the following holds:
\begin{enumerate}
    \item There exist edges $ (u', u), (v', v) \in E $ such that $ u' ~\|~ v' $;
    \item There exist edges $ (u', u), (v', v), (u'', u), (v'', v) \in E $ such that $ u' < v'$ and $ v'' < u'' $.
\end{enumerate}

Indeed, if e.g. it were $ u < v $, then Axiom 2 would imply that in case 1 it should hold $ u' \le v' $ and in case 2 it should hold $ u'' \le v'' $ (which is forbidden by antisymmetry of $\leq$).
\end{remark}

Figure \ref{fig:2-sortable} depicts the running example that will be used throughout the paper. The automaton recognizes the regular language $\mathcal L = ab(aa)^*(bb)^*$. It can be shown (see \cite{alanko2020wheeler}) that $\mathcal L$ cannot be recognized by any Wheeler automaton. Consider the co-lexicographic order $ \le $ whose Hasse diagram is depicted in the figure.
The partial order's width is 2, and the order can be partitioned into two $ \le $-chains: this will become important later in the paper. 
The right part of the figure makes it clear that a possible $ \le $-chain decomposition (not the only one) is $\{\{0,1,3,6\}$, $\{4,2,5\}\}$.

\begin{figure}[h!]
\centering
\begin{subfigure}{.45\textwidth}
	\centering
	\begin{tikzpicture}[shorten >=1pt,node distance=1.6cm,on grid,auto]
	\tikzstyle{every state}=[fill={rgb:black,1;white,10}]
	
	\node[state,initial]   (q_0)                    {$0$};
	\node[state]           (q_1)  [right of=q_0]    {$1$};
	\node[state]           (q_3)  [right of=q_1]    {$3$};
	\node[state,accepting] (q_4)  [right of=q_3]    {$4$};
	\node[state,accepting] (q_2)  [below of=q_3]    {$2$};	
	\node[state]           (q_5)  [right of=q_2]    {$5$};
	\node[state,accepting] (q_6)  [right of=q_5]    {$6$};

	\path[->]
	(q_0) edge node {a}    (q_1)
	(q_1) edge node {b}    (q_2)
	(q_2) edge node {a}    (q_3)
	(q_2) edge node {b}    (q_5)
	(q_4) edge node {b}    (q_5)
	(q_4) edge [bend left] node {a}    (q_3)
	(q_3) edge [bend left] node {a}    (q_4)
	(q_5) edge [bend left] node {b}    (q_6)
	(q_6) edge [bend left] node {b}    (q_5);
	\end{tikzpicture}
\end{subfigure}
\begin{subfigure}{.45\textwidth}
	\centering
	\begin{tikzpicture}[shorten >=1pt,node distance=0.7cm,on grid,auto]
	\tikzstyle{every state}=[fill={rgb:black,1;white,10}]
	
	\node[state,color=white,text=black,inner sep=1pt,minimum size=0pt] (q_0)                   {$0$};
	\node[state,color=white,text=black,inner sep=1pt,minimum size=0pt] (q_1)  [above of=q_0]   {$1$};
	\node[state,color=white,text=black,inner sep=1pt,minimum size=0pt] (q_3)  [above of=q_1]   {$3$};
	\node[state,color=white,text=black,inner sep=1pt,minimum size=0pt] (q_4)  [right of=q_3]   {$4$};
	\node[state,color=white,text=black,inner sep=1pt,minimum size=0pt] (q_2)  [above of=q_4]   {$2$};	
	\node[state,color=white,text=black,inner sep=1pt,minimum size=0pt] (q_5)  [above of=q_2]   {$5$};
	\node[state,color=white,text=black,inner sep=1pt,minimum size=0pt] (q_6)  [left of=q_5]    {$6$};
	
	\path[->]
	(q_0) edge node {}    (q_1)
	(q_1) edge node {}    (q_3)
	(q_3) edge node {}    (q_6)
	(q_4) edge node {}    (q_2)
	(q_2) edge node {}    (q_5)
	(q_2) edge node {}    (q_6)
	(q_3) edge node {}    (q_2)
	(q_1) edge node {}    (q_4)
	(q_0) edge [bend right] node {}    (q_4);
	\end{tikzpicture}
\end{subfigure}
\caption{\textbf{Left}: automaton recognizing the non-Wheeler language $\mathcal L = ab(aa)^*(bb)^*$.  \textbf{Right}: Hasse diagram of a co-lexicographic partial order of the states.} \label{fig:2-sortable}
\end{figure}
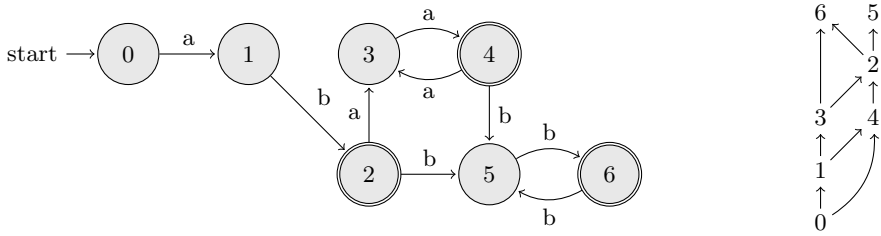

The class of Wheeler languages --- that is, the class of all regular languages recognized by some Wheeler NFA --- is rather small: for example, unary languages are Wheeler only if they are finite or co-finite. In general, Wheeler languages are not closed under union, complement, concatenation, and Kleene star \cite{alanko2020wheeler}.
On the other hand, every finite automaton admits a co-lexicographic order, because:
\begin{equation*}
\leq\ :=\ \{(u, u) \in Q \times Q\ |\ u \in Q \} \cup \{(u, v) \in Q \times Q\ |\ \lambda(u) < \lambda(v) \}.
\end{equation*}
satisfies Axiom 1 and trivially satisfies Axiom 2. This order, however, in general is not "complete" in the sense that it can possibly be expanded with new pairs. Later we will define a maximality criterion that a co-lexicographic order should have in order to be useful for indexing and compression.

\begin{definition}\label{definition2}
A \emph{co-lexicographic nondeterministic (resp. deterministic) finite automaton (CNFA, resp. CDFA)} is a 6-tuple $ (Q, E, \Sigma, s, F, \le) $ where $ (Q, E, \Sigma, s, F) $ is an NFA (resp. DFA) and $ \le $ is a co-lexicographic order of the automaton.
\end{definition}

Following the notation of Alanko et al. \cite{alanko2020wheeler,alanko20regular}, let $ \mathcal{L}({\mathcal{A}}) $ be the language accepted by $ \mathcal{A} $ and let $ Pref (\mathcal{L}({\mathcal{A}})) $ be the set of all strings in $ \Sigma^* $ that can be read on $\mathcal A$ by following some path starting from the initial state $ s $. For any $ \alpha \in \Sigma^* $, we denote by $ I_\alpha $ the set of all states that can be reached from the initial state $ s $ by following a path whose edges, when concatenated, yield $ \alpha $.

The following Lemma exhibits the nature of our ordering $\leq$ among states. 
Intuitively, we prove that $\leq$ must respect the co-lexicographic order of the strings that can be read from the initial state.

\begin{lemma}[Generalized from \cite{alanko20regular}]\label{lem:string-nodes}
Let $ \mathcal{A} = (Q, E, \Sigma, s, F, \le) $ be a CNFA. Let $ u, v \in Q $ and $ \alpha, \beta, \in Pref (\mathcal{L(A)}) $ such that $ u \in I_\alpha $, $ v \in I_\beta $ and $ \{u, v \} \not \subseteq I_\alpha \cap I_\beta $.
\begin{enumerate}
\item If $ \alpha < \beta $, then $ u ~\|~ v $ or $ u < v $.
\item If $ u < v $, then $ \alpha \prec \beta $.
\end{enumerate}
\end{lemma}

The reason why Wheeler automata admit an efficient indexing mechanism lies in two key observations: (i) the set of states reached by a path labeled with a given string $\alpha$ forms \emph{an interval}, and (ii) on total orders an interval can be expressed with $O(1)$ words by specifying its endpoints. We now show that the interval property holds true also for partial orders.

\begin{lemma}[Path coherence, generalized from \cite{GAGIE201767}]\label{lem5}
Let $ \mathcal{A} = (Q, E, \Sigma, s, F, \le) $ be a CNFA. Let $ \alpha \in \Sigma^*$, and let $ U $ be a $ \le $-interval of states. Then, the set $ U' $ of all states in $ Q $ that can be reached from $ U $ by following edges whose labels, when concatenated, yield $ \alpha $, is still a $ \le $-interval.
\end{lemma}

\begin{corollary}\label{cor:1}
Let $ \mathcal{A} = (Q, E, \Sigma, s, F, \le) $ be a CNFA. Let $ \alpha \in \Sigma^*$. Then, $ I_\alpha $ is a $ \le $-interval.
\end{corollary}

\begin{proof}
Pick $ U = \{s\} $ in Lemma \ref{lem5}. Then $ U' = I_\alpha $. For an alternative proof, see Appendix \ref{app:cor1}. \qed
\end{proof}

As we will see, the above results imply that indexing can be extended to arbitrary finite automata by updating \emph{one} $\leq$-interval for each character of the query pattern. This however does not mean that, in general, indexing can be performed efficiently as on Wheeler automata: as we show next, in general a $\leq$-interval cannot be expressed in constant space.

\begin{lemma}\label{lem:lower bound}
The following hold: (1) Any partial order $(V,\leq)$ of width $ p $  has at least $2^p$ distinct $\leq$-intervals.
(2) For any size $n=|V|$ and any $1\leq p \leq n$, there exists a partial order $(V,\leq)$ of width $p$ with at least $(n/p)^p$ distinct $\leq$-intervals.
\end{lemma}
\begin{proof}[Sketch] The general bound $2^p$ follows immediately from the fact that any subset of an antichain is a $\leq$-interval. For the bound $(n/p)^p$, consider an order composed by $p$ mutually-incomparable chains. Every combination of intervals on the chains forms a distinct $\leq$-interval. \qed
\end{proof}

\begin{remark}\label{rem:lower bounds bit}
Given an NFA $\mathcal A$ with $n$ states and a co-lexicographic order $\leq$ of width $p$ of $\mathcal A$, Lemma \ref{lem:lower bound} (1) implies an information-theoretic lower bound of $p$ bits for expressing a $\leq$-interval. By Lemma \ref{lem:lower bound} (2), this bound increases to $\Omega(p\log(n/p))$ bits in the worst case. This means that, up to (possibly) a logarithmic factor, in the word RAM model $\Omega(p)$ time is needed to manipulate one $\leq$-interval.
\end{remark}

The above remark motivates the following strategy. Letting $p$ be the width of a partial order $ \le $, by Dilworth's theorem \cite{dilworth2009decomposition} there exists a $ \le $-chain decomposition $ \{Q_i\}_{i=1}^p $ of $Q$ into $p$ chains. Then, the following lemma implies that a $ \le $-interval can be encoded by at most $p$ intervals, each contained in a distinct chain, using $O(p)$ words. This encoding is essentially optimal by Remark \ref{rem:lower bounds bit}. 

\begin{lemma}
Let $ (V, \le) $ be a partial order, and let $ U $ be a $ \le $-interval. Let $ \{V_i\}_{i = 1}^p $ be a $ \le $-chain decomposition of $ V $. Then, $ U $ is the disjoint union of $ p $ (possibly empty) sets $ U_1, \dots, U_p $, where $ U_i $ is a $ \le_{V_i} $-interval, for $ i = 1, \dots, p $.
\end{lemma}

\begin{proof}
Define $ U_i := U \cap V_i $. Then, $ U $ is the disjoint union of all $ U_i $'s, because $ \{V_i\}_{i = 1}^p $ is a partition. Hence we just have to prove that $ U_i $ is a $ \le_{V_i} $-interval.
Let $ u, v, z \in V_i $ be such that $ u < v < z $ and $ u, z \in U_i $. In particular $ u, z \in U $, so $ v \in U $ (because $ U $ is a $ \le $-interval) and we conclude $ v \in U_i $. \qed
\end{proof}

We can now restate Lemma \ref{lem5} and Corollary \ref{cor:1} as follows.

\begin{lemma}\label{lem:path coherence - chains}
Let $ \mathcal{A} = (Q, E, \Sigma, s, F, \le) $ be a CNFA, and let $ \{Q_i \}_{i = 1}^p $ be a $ \le $-chain decomposition of $ Q $. Let $ \alpha \in \Sigma^*$, and let $ U $ be a $ \le $-interval of states. Then, the set $ U' $ of all states in $ Q $ that can be reached from $ U $ by following edges whose labels, when concatenated, yield $ \alpha $, is the disjoint union of $ p $ (possibly empty) sets $ U'_1, \dots, U'_p $, where $ U'_i $ is a $ \le_{Q_i} $-interval, for $ i = 1, \dots, p $.
\end{lemma}

\begin{corollary}\label{cor:path-coherence}
Let $ \mathcal{A} = (Q, E, \Sigma, s, F, \le) $ be a CNFA, and let $ \{Q_i \}_{i = 1}^p $ be a $ \le $-chain decomposition of $ Q $. Let $ \alpha \in Pref (\mathcal{L(A)}) $. Then, $ I_\alpha $ is the disjoint union of $ p $ (possibly empty) sets $ I_\alpha^1, \dots, I_\alpha^p $, where $ I_\alpha^i $ is a $ \le_{Q_i} $-interval, for $ i = 1, \dots, p $.
\end{corollary}

Lemma \ref{lem:path coherence - chains} stands at the core of the indexing results of Section \ref{sec:idx}, 
where we will also show that the properties of co-lexicographic orders allow storing the automaton in $O(\log p)$ bits per transition on top of the labels. This motivates the problem of minimizing $p$, considered in Section \ref{sec:minimize p}.

\begin{definition}
Let $ \mathcal{A} = (Q, E, \Sigma, s, F) $ be an NFA.
\begin{enumerate}
\item We say that $ \mathcal{A} $ is \emph{$ p $-sortable} if there exists a co-lexicographic order $ \le $ of $ \mathcal{A} $ such that $ Q $ admits a $ \le $-chain decomposition $ \{Q_i\}_{i = 1}^p $.
\item The \emph{co-lexicographic width} $ \bar p $ of $\mathcal A$ is the smallest integer $ p $ for which $ \mathcal{A} $ is $ p $-sortable.
\end{enumerate}
\end{definition}

An NFA is Wheeler \cite{GAGIE201767} if and only if it is 1-sortable, i.e. if it has co-lexicographic width $\bar p = 1$.

\section{Compression and Indexing}\label{sec:idx}

The Burrows-Wheeler transform (BWT) \cite{burrows1994block} of a text is a permutation that re-arranges the text's characters according to the co-lexicographic order of the prefixes that precede them. 
The BWT boosts compression and enables efficient indexing in compressed space \cite{ferragina2000opportunistic}.
Previous works generalized this transform to trees \cite{Ferragina2005}, string sets \cite{mantaci2007extension}, de Bruijn graphs \cite{BOSS,GCSA} and Wheeler graphs \cite{GAGIE201767}.
In this section we finally generalize the BWT to arbitrary finite automata.

In order to introduce our transform, let us consider the example of Figure \ref{fig:2-sortable}. Consider the chain decomposition $Q_1 = \{0,1,3,6\}$, $Q_2 = \{4,2,5\}$ and the sequence of states obtained by concatenating the sorted elements of $Q_1$ and $Q_2$: $0,1,3,6, 4,2,5$. The left part of Figure \ref{tab:adj} visualizes the adjacency matrix of the graph using this state order. 
Chains $Q_1$ and $Q_2$ are highlighted in blue and yellow, respectively.
Partition the adjacency matrix into blocks by drawing a horizontal line every time a new chain starts and a vertical line every time a new chain starts or the label associated with the considered state changes. In our example, we obtain 10 blocks, shown in the left part of Figure \ref{tab:adj} using different shades of gray. It is not hard to see that, by Definition \ref{def:co-lex order}, blocks are monotone, i.e. equally-labeled pairs of edges leaving a chain $Q_i$ and landing inside a chain $Q_j$ (possibly, $i=j$) preserve the co-lexicographic order of their endpoint states. As we observe below, this makes it possible to compress the matrix because for each edge we only need to specify its label and the two endpoint chains. The right part of Figure \ref{tab:adj} shows this construction.
This is a two-dimensional visualization of the Burrows-Wheeler transform of the NFA, which can be linearized in two sequences ($\tt OUT$ and $\tt IN$) as shown in the table. 

\definecolor{gr1}{rgb}{0.78, 0.78, 0.78}
\definecolor{gr2}{rgb}{0.9, 0.9, 0.9}

\definecolor{ce1}{rgb}{0.61, 0.77, 0.89}
\definecolor{pe1}{rgb}{1.0, 0.9, 0.71}
\definecolor{ao}{rgb}{0.0, 0.5, 0.0}

\newcolumntype{P}[1]{>{\centering\arraybackslash}p{#1}}

\begin{figure}
  \centering
  \begin{tabular}{|P{10pt}|P{10pt}|P{10pt}|P{10pt}|P{10pt}|P{10pt}|P{10pt}|P{10pt}|}
  \cline{2-8}
    \multicolumn{1}{c|}{} & $\#$ & \multicolumn{2}{c|}{a}&b&a&\multicolumn{2}{c|}{b}\\\cline{2-8}
    \multicolumn{1}{c|}{} & \cellcolor{ce1} 0 & \cellcolor{ce1} 1 & \cellcolor{ce1} 3 & \cellcolor{ce1} 6 & \cellcolor{pe1} 4 & \cellcolor{pe1} 2 & \cellcolor{pe1} 5 \\\hline
     \cellcolor{ce1} 0 & \cellcolor{gr1} & \cellcolor{gr2} a & \cellcolor{gr2} & \cellcolor{gr1} & \cellcolor{gr2} & \cellcolor{gr1} & \cellcolor{gr1} \\\hline
     \cellcolor{ce1} 1 & \cellcolor{gr1} & \cellcolor{gr2} & \cellcolor{gr2} & \cellcolor{gr1} & \cellcolor{gr2} & \cellcolor{gr1} b & \cellcolor{gr1} \\\hline
     \cellcolor{ce1} 3 & \cellcolor{gr1} & \cellcolor{gr2} & \cellcolor{gr2} & \cellcolor{gr1} & \cellcolor{gr2} a & \cellcolor{gr1} & \cellcolor{gr1} \\\hline
     \cellcolor{ce1} 6 & \cellcolor{gr1} & \cellcolor{gr2} & \cellcolor{gr2} & \cellcolor{gr1} & \cellcolor{gr2} & \cellcolor{gr1} & \cellcolor{gr1} b \\\hline
    \cellcolor{pe1} 4 & \cellcolor{gr2} & \cellcolor{gr1} & \cellcolor{gr1} a & \cellcolor{gr2} & \cellcolor{gr1} & \cellcolor{gr2} & \cellcolor{gr2} b\\\hline
     \cellcolor{pe1} 2 & \cellcolor{gr2} & \cellcolor{gr1} & \cellcolor{gr1} a & \cellcolor{gr2} & \cellcolor{gr1} & \cellcolor{gr2} & \cellcolor{gr2} b \\\hline
     \cellcolor{pe1} 5 & \cellcolor{gr2} & \cellcolor{gr1} & \cellcolor{gr1} & \cellcolor{gr2} b & \cellcolor{gr1} & \cellcolor{gr2} & \cellcolor{gr2} \\\hline
  \end{tabular}
  \hspace{20pt}
   \begin{tabular}{r|P{10pt}|P{30pt}|P{30pt}|P{30pt}|P{30pt}|P{30pt}|P{30pt}|P{30pt}|}
    \multicolumn{1}{c}{} & \multicolumn{1}{c}{\textcolor{ao}{$\tt{IN}$}} & \multicolumn{1}{c}{\textcolor{ao}{[]}} & \multicolumn{1}{c}{\textcolor{ao}{[1]}} & \multicolumn{1}{c}{\textcolor{ao}{[2,2]}} & \multicolumn{1}{c}{\textcolor{ao}{[2]}} & \multicolumn{1}{c}{\textcolor{ao}{[1]}} & \multicolumn{1}{c}{\textcolor{ao}{[1]}} & \multicolumn{1}{c}{\textcolor{ao}{[1,2,2]}}\\\cline{3-9}
    \multicolumn{1}{c}{\textcolor{red}{$\tt{OUT}$}} & \multicolumn{1}{c|}{} & \cellcolor{ce1} 0 & \cellcolor{ce1} 1 & \cellcolor{ce1} 3 & \cellcolor{ce1} 6 & \cellcolor{pe1} 4 & \cellcolor{pe1} 2 & \cellcolor{pe1} 5 \\\cline{2-9}
     \multicolumn{1}{c|}{\textcolor{red}{[(1,a)]}} & \cellcolor{ce1} 0 & \cellcolor{gr1} & \cellcolor{gr2} (\textcolor{ao}1,\textcolor{red}{1,a}) & \cellcolor{gr2} & \cellcolor{gr1} & \cellcolor{gr2} & \cellcolor{gr1} & \cellcolor{gr1} \\\cline{2-9}
     \multicolumn{1}{c|}{\textcolor{red}{[(2,b)]}} & \cellcolor{ce1} 1 & \cellcolor{gr1} & \cellcolor{gr2} & \cellcolor{gr2} & \cellcolor{gr1} & \cellcolor{gr2} & \cellcolor{gr1} (\textcolor{ao}1,\textcolor{red}{2,b}) & \cellcolor{gr1} \\\cline{2-9}
     \multicolumn{1}{c|}{\textcolor{red}{[(2,a)]}} & \cellcolor{ce1} 3 & \cellcolor{gr1} & \cellcolor{gr2} & \cellcolor{gr2} & \cellcolor{gr1} & \cellcolor{gr2} (\textcolor{ao}1,\textcolor{red}{2,a}) & \cellcolor{gr1} & \cellcolor{gr1} \\\cline{2-9}
     \multicolumn{1}{c|}{\textcolor{red}{[(2,b)]}} & \cellcolor{ce1} 6 & \cellcolor{gr1} & \cellcolor{gr2} & \cellcolor{gr2} & \cellcolor{gr1} & \cellcolor{gr2} & \cellcolor{gr1} & \cellcolor{gr1} (\textcolor{ao}1,\textcolor{red}{2,b}) \\\cline{2-9}
    \multicolumn{1}{c|}{\textcolor{red}{[(1,a),(2,b)]}} & \cellcolor{pe1} 4 & \cellcolor{gr2} & \cellcolor{gr1} & \cellcolor{gr1} (\textcolor{ao}2,\textcolor{red}{1,a}) & \cellcolor{gr2} & \cellcolor{gr1} & \cellcolor{gr2} & \cellcolor{gr2} (\textcolor{ao}2,\textcolor{red}{2,b})\\\cline{2-9}
     \multicolumn{1}{c|}{\textcolor{red}{[(1,a),(2,b)]}} & \cellcolor{pe1} 2 & \cellcolor{gr2} & \cellcolor{gr1} & \cellcolor{gr1} (\textcolor{ao}2,\textcolor{red}{1,a}) & \cellcolor{gr2} & \cellcolor{gr1} & \cellcolor{gr2} & \cellcolor{gr2} (\textcolor{ao}2,\textcolor{red}{2,b}) \\\cline{2-9}
     \multicolumn{1}{c|}{\textcolor{red}{[(1,b)]}} & \cellcolor{pe1} 5 & \cellcolor{gr2} & \cellcolor{gr1} & \cellcolor{gr1} & \cellcolor{gr2} (\textcolor{ao}2,\textcolor{red}{1,b}) & \cellcolor{gr1} & \cellcolor{gr2} & \cellcolor{gr2} \\\cline{2-9}
  \end{tabular}
  \vspace{5pt}
  \caption{\textbf{Left}. Adjacency matrix of the finite automaton of Figure \ref{fig:2-sortable}. The different shades of gray highlight the monotone blocks.
  \textbf{Right}. Inside the gray cells: two-dimensional visualization of the Burrows-Wheeler transform of the automaton. For each edge we need to store only the two endpoint chains and the label. The BWT can be linearized in two sequences \texttt{OUT} and \texttt{IN}: collect \emph{vertically} the first component of every triple (in green) and \emph{horizontally} the other two components (in red). The transition function can be reconstructed from $\tt{OUT}$ and $\tt{IN}$.
  Square brackets indicate lists.
  }
  
  \label{tab:adj}
\end{figure}

\begin{definition}[BWT of an NFA]
Let $ \mathcal{A} = (Q, E, \Sigma, s, F) $ be an NFA. Let $ \le $ be a co-lexicographic order of $ \mathcal{A }$, and let $\mathcal Q := \{Q_i\}_{i=1}^p $ be a $ \le $-chain decomposition of $ Q $, with $s\in Q_1$ without loss of generality. Let $ \pi(v) $, with $v\in Q$, denote the unique integer such that  $ v \in Q_{\pi(v)} $.
Consider the ordering $v_1, \dots, v_n $ of $ Q $ 
such that for every $ 1 \le i < j \le n $ it holds  $ \pi(v_i) < \pi(v_j) \lor (\pi(v_i) = \pi(v_j) \land v_i < v_j) $. The BWT of $(\mathcal A, \leq, \mathcal Q)$ is the triple of sequences $\tt{BWT} = (\tt{OUT}, \tt{IN},  FINAL)$, each of length $ n $, such that, for every $ i = 1, \dots, n $:
\begin{itemize}
    \item $\mathtt{OUT}[i]$ is the list of all pairs $(\pi(u),c)$, for every edge $(v_i,u,c) \in E $ leaving $ v_i $.
    \item $\mathtt{IN}[i]$ is the list of all integers $\pi(w)$, for every edge $(w,v_i,c) \in E $ reaching $ v_i $.
    \item $\mathtt{FINAL}[i] = 1$ if $v_i\in F$, and $0$ otherwise. 
\end{itemize}
\end{definition}

To simplify our exposition, we will simply say \emph{BWT of a $p$-sortable NFA $\mathcal A$} to indicate the BWT of $(\mathcal A, \leq, \mathcal Q)$, for some given co-lexicographic order $\leq$ of $ \mathcal{A} $ and some given $ \le $-chain decomposition $\mathcal Q = \{Q_i\}_{i=1}^p $ of $Q$.
It is not hard to see that our BWT generalizes all existing approaches \cite{burrows1994block,Ferragina2005,mantaci2007extension,BOSS,GCSA,GAGIE201767}, for which $p=1$ always holds. For example, on (circular) strings the integers $\pi(u)$ and $\pi(w)$ are always equal to 1 and the lists $\mathtt{OUT}[i]$ and $\mathtt{IN}[i]$ have length 1. After removing the uninformative integers $\pi(u)$ and $\pi(w)$, only one label per state is left and sequence $\tt{OUT}$ coincides with the classic BWT. Similarly, also Wheeler automata satisfy $p=1$; 
in this case, however, the information provided by $\pi(u)$ and $\pi(w)$ (in particular, the  lengths of the lists containing them) must be kept in order to encode the in- and out- degrees of the states.

\begin{theorem}\label{th:store NFA}
The BWT of a $p$-sortable NFA $ \mathcal{A} = (Q, E, \Sigma, s, F) $ can be stored as an invertible representation using $|E|(\lceil \log \sigma \rceil + 2\lceil\log p\rceil + 2) + |Q|$ bits.
\end{theorem}

We now extend the FM-index \cite{ferragina2000opportunistic} to arbitrary NFAs. Note that the index is more space-efficient than the encoding of Theorem \ref{th:store NFA} by an additive term $\lceil \log p\rceil - 2$. There is a deep reason for this fact: as we show next, within this space we can compute any interval $I_\alpha$ (more in general, any $\leq$-interval of states reached by some string) by forward search. However, in general on NFAs $I_\alpha$ is not a singleton. As a result, states in $I_\alpha$ cannot be distinguished and forward search does not permit to invert the automaton. On the other hand, in Theorem \ref{th:BWT of DFA} we will observe that DFAs can be inverted within this space precisely because, in that case, the intervals $I_\alpha$ are singletons.

\begin{theorem}[FM-index of an NFA]\label{th:FM-index}
The BWT of a $p$-sortable NFA $ \mathcal{A} = (Q, E, \Sigma, s, F) $ can be encoded with a data structure of $|E|(\lceil \log\sigma\rceil + \lceil\log p\rceil + 2)\cdot (1+o(1)) + 2|Q|\cdot (1+o(1))$ bits that, given a query string $\alpha\in \Sigma^m$,  supports the following operations in $O(m\cdot p^2 \cdot \log(p\cdot \sigma))$ time:  
\begin{itemize}
    \item[(i)] Count the number of states reached by a path labeled $\alpha$.
    \item[(ii)] Return unique identifiers ($p$ ranges on the $p$ chains) for the states reached by a path labeled $\alpha$.
    \item[(iii)] Decide whether $\alpha \in \mathcal L(\mathcal A)$.
\end{itemize}
Given any $I_\alpha$ and $c\in \Sigma$, the structure can furthermore compute $I_{\alpha\cdot c}$ in $O(p^2 \cdot \log(p\cdot \sigma))$ time.
\end{theorem}

In the worst case ($p=n$ and $|E| = \Omega(n^2)$) our index supports pattern matching on NFAs in time $O(m\cdot |E| \cdot \log n)$. This bound matches the lower bound of Equi et al. \cite{equi2020conditional,equi2020graphs} up to a logarithmic factor. To conclude, we observe that forward search enables a more efficient encoding for DFAs.

\begin{theorem}\label{th:BWT of DFA}
The BWT of a $p$-sortable DFA $ \mathcal{A} = (Q, E, \Sigma, s, F) $ can be stored as an invertible representation using $|E|(\lceil \log \sigma \rceil + \lceil\log p\rceil + 2) + |Q|$ bits.
\end{theorem}

\section{Relating the co-lexicographic width with NFA determinization}

In this section we prove another surprising result related to the co-lexicographic width $\bar p$: the classic powerset construction algorithm for NFA determinization generates an equivalent DFA of size exponential in $\bar p$, rather than in the  number of the NFA's states.

Let $ \mathcal{A} = (Q, E, \Sigma, s, F) $ be an NFA. The powerset construction algorithm builds a DFA $ \mathcal{A^*} = (Q^*, E^*, \Sigma, s^*, F^*) $ such that $ \mathcal{L(\mathcal{A})} = \mathcal{L(\mathcal{A^*})} $ defined as: 
(i) $ Q^* = \{ I_{\alpha}\ |\ \alpha\in Pref (\mathcal{L(\mathcal{A})})\}$, 
(ii) $ E^* = \{(I_\alpha, I_{\alpha e}, e)\ |\ \alpha \in \Sigma^*, e \in \Sigma, \alpha e \in Pref (\mathcal{L(\mathcal{A})}) \} $, (iii) $ s^* = \{s \} $, and (iv) $ F^* = \{I_\alpha\ |\ \alpha \in \mathcal{L(\mathcal{A})} \} $.

The number of states in $ Q^* $ can be exponential in the number of states in $ Q $, because in principle every nonempty subset of $ Q $ may be a state in $ Q^* $. On the other hand, Corollary \ref{cor:path-coherence} implies that, if $ \mathcal{A} $ is $ p $-sortable, then $ |Q^*| $ must be somehow bounded, because the elements of $ Q^* $ cannot be arbitrary as they must be the union of at most $ p $ intervals. Since every automaton admits a co-lexicographic order, we expect the co-lexicographic width to be a measure of the growth of $ |Q^*| $. For space reasons, our analysis has been moved to Appendix \ref{app:powerset}. We obtain the following result:

\begin{theorem}\label{th:powerset}
Let $ \mathcal{A} = (Q, E, \Sigma, s, F) $ be a $ p $-sortable NFA, and let $ \mathcal{A^*} = (Q^*, E^*, \Sigma, s^*, F^*) $ be the DFA obtained from $ \mathcal{A} $ by the powerset construction algorithm. Then, $ |Q^*| \le 2^p (|Q| - p + 1) - 1 $.
\end{theorem}

Clearly, Theorem \ref{th:powerset} holds also for the co-lexicographic width $\bar p$ of $\mathcal{A}$. The theorem has an intriguing consequence: the PSPACE-complete NFA equivalence problem \cite{stockmeyer1973word} is fixed-parameter tractable with respect to $\bar p$. In order to prove this result, we first update the analysis of Hopcroft et al. \cite{hopcroft2006introduction} of the powerset construction algorithm.

\begin{lemma}[Adapted from \cite{hopcroft2006introduction}]\label{lem:powerset complexity}
Given a $p$-sortable NFA with $n$ states on alphabet of size $\sigma$, in $O(2^p(n-p+1)n^2\sigma)$ time the powerset construction algorithm generates an equivalent DFA with at most $2^p(n-p+1)$ states.
\end{lemma}
\begin{proof}
From Theorem \ref{th:powerset}, let $N = 2^p(n-p+1)$ be an upper-bound to the number of states of the equivalent DFA. Each DFA state $x$ is a set $x=\{u_1, \dots, u_k\}$ formed by $k \leq n$ states $u_1, \dots, u_k$ of the original NFA. For each character $c\in \Sigma$ labeling an edge leaving $x$, we need to follow all edges labeled $c$ from $u_1, \dots, u_k$. In the worst case (a complete transition function), this leads to traversing $O(k\cdot n) \subseteq O(n^2)$ edges of the NFA. The final complexity is thus  $O(N\cdot n^2\cdot \sigma)$. \qed
\end{proof}

\begin{corollary}\label{cor:NFA equiv}
We can check the equivalence between two $p$-sortable NFAs with at most $n$ states each over alphabet of size $\sigma$ in $O(2^p(n-p+1)n^2\sigma)$ time. 
\end{corollary}
\begin{proof}
Simply build the equivalent DFAs, of size at most $N\leq 2^p(n-p+1)$, by powerset construction. 
Finally, DFA equivalence can be tested in $O(N\sigma\log N)$ time by DFA minimization using Hopcroft's algorithm. The final running time is dominated by powerset construction, see Lemma  \ref{lem:powerset complexity}. \qed 
\end{proof}

Similarly, NFA determinization can be used to test membership of a word in a regular language expressed as an NFA. For sufficiently small $\bar p$ and $n$, this solution is faster than the classic one running in $O(mn)$ time based on dynamic programming \cite{thompson1968regular}:

\begin{corollary}
We can test membership of a word of length $m$ in the language recognized by a $p$-sortable NFA with $n$ states on alphabet of size $\sigma$ in $O(2^p(n-p+1)n^2\sigma + m)$ time. 
\end{corollary}

\section{Complexity Results}\label{sec:minimize p}

In the previous sections we have seen that the co-lexicographic width $ \bar{p} $ is a relevant complexity measure for several problems on finite automata. In order to make our indexing and compression results of practical value, the next step is to study the problem of finding or bounding $ \bar{p} $ and determining a corresponding chain decomposition of the automaton's states.

We define the decision version of the \emph{sortability problem} as follows: given an NFA $ \mathcal{A} $ and an integer $ p $, determine whether $ \mathcal{A} $ is $ p $-sortable. 

\begin{theorem}\label{th:NP hard}
The sortability problem is NP-hard.
\end{theorem}

\begin{proof}
A graph is Wheeler if and only it is $ 1 $-sortable. The conclusion follows by the NP-completeness of the problem of recognizing whether an NFA is Wheeler \cite{gibney19}. \qed
\end{proof}

In principle, we can obtain an upper bound to the co-lexicographic width by determining a co-lexicographic order $ \le $ and then building a $ \le $-chain decomposition. The hardness of the sortability problem lays in the hardness of finding the co-lexicographic order of smallest width, because, fixed a co-lexicographic order $ \le $, we can find its $ \le $-width (and an associated $ \le $-chain decomposition) in polynomial time:

\begin{lemma}\label{lem:complexity min chain partition}
Let $(V,\leq)$ be a partial order, with $|V|=n$. The smallest $ \le $-chain decomposition of $V$ can be found in $O(n^{5/2})$ time. 
\end{lemma}
\begin{proof}
Ford and Fulkerson \cite{ford1962flows} provided a reduction from the minimum chain decomposition problem to the maximum matching problem on the bipartite graph $(V',V'',E)$, where $V'=V''=V$ and $(v',v'')\in E \subseteq V'\times V''$ iff $v'\leq v''$. By completing the matching with pairs $(v,v) \in V'\times V''$ for every $v\in V$, the resulting connected components are the foresought chains. The complexity of the procedure thus reduces to that of finding a maximum matching, which can be solved by Hopcroft and Karp's algorithm \cite{HKmatching} in  $O(|E|\cdot \sqrt{n}) \in O(n^{5/2})$ time. \qed
\end{proof}

Let us study the family of  co-lexicographic orders that can be defined on the states of an NFA.

\begin{definition}
Let $ \mathcal{A} = (Q, E, \Sigma, s, F) $ be an NFA.
\begin{enumerate}
\item Let $ \le, \le^* $ be co-lexicographic orders on $ Q $. We say that $ \le^* $ is a \emph{refinement} of $ \le $ if:
\begin{equation*}
    u \le v \implies u \le^* v \quad \text{ $ \forall  u, v \in Q $}.
\end{equation*}
\item A co-lexicographic order $ \le $ on $ Q $ is \emph{maximal} if its unique refinement is $ \le $ itself.
\end{enumerate}
\end{definition}

\begin{remark}
Every co-lexicographic order $ \le $ is refined by a maximal co-lexicographic order. Indeed, either $ \le $ is maximal or $ \le $ is refined by some other co-lexicographic order, so we can build a non-extendable chain of pairwise distinct co-lexicographic orders such that every co-lexicographic order refines the previous one. Clearly, such a chain must be finite, so we obtain a maximal co-lexicographic order that refines $ \le $.
\end{remark}

\begin{remark}\label{rem:w*<w}
Assume that $ \le^* $ is a refinement of $ \le $, and let $ w^*$ and $ w $ be their corresponding widths. Then, it must be $ w^* \le w $ because every $ \le $-chain decomposition is also a $ \le^* $-chain decomposition. This implies that if $\bar p $ is the co-lexicographic width of $ \mathcal{A} = (Q, E, \Sigma, s, F) $, then there exists a \emph{maximal} co-lexicographic order on $ Q $ whose width is $\bar p $.
\end{remark}

In general, an NFA admits several maximal co-lexicographic orders (see also \cite{alanko20regular} for the case $\bar p=1$):
consider, as a simple example, a source $s$ connected by the same label to $n$ pairwise not-adjacent states.
Notably, in Appendix \ref{app:DFA unique} we prove  that DFAs admit a unique maximal co-lexicographic order. Moreover, we show that such an order can be found in polynomial time:

\begin{theorem}\label{th:complexity co-lex DFA}
Let $ \mathcal{A} = (Q, E, \Sigma, s, F) $ be a DFA. We can find the unique maximal co-lexicographic order of $ \mathcal{A} $ in $ O(|E|^2) $ time.
\end{theorem}
\begin{proof}[Sketch]
We show that, by sorting any spanning tree of $\mathcal A$ rooted in $s$, we obtain a total order $\leq_{\#}$ which is a superset of the unique maximal order $\leq^*$ of $\mathcal A$. The order $\leq^*$ is obtained by (i) finding "base-case" incomparable state pairs having inconsistent predecessors by $\leq_{\#}$, and (ii) propagating the incomparability relation by following pairs of equally-labeled edges. \qed
\end{proof}

\begin{remark}\label{rem3}
If there exists only one maximal co-lexicographic order $ \le ^* $, then $ \le^* $ refines every co-lexicographic order, and by Remark \ref{rem:w*<w} the $ \le^* $-width of $ Q $ is the co-lexicographic width $\bar p$ of $ \mathcal{A} $.
\end{remark}

By Theorem \ref{th:complexity co-lex DFA}, Remark \ref{rem3}, and Lemma \ref{lem:complexity min chain partition} we obtain:

\begin{corollary}\label{cor:poly-DFA}
Let $ \mathcal{A} = (Q, E, \Sigma, s, F) $ be a DFA. We can find the unique maximal co-lexicographic order $\leq$ of $\mathcal A$ and the corresponding smallest chain decomposition $ \{Q_i \}_{i = 1}^{\bar{p}} $, where  $\bar p$ is the co-lexicographic width of $ \mathcal{A}$, in $O(|E|^2+|Q|^{5/2}) $ time.
\end{corollary}

Corollary \ref{cor:poly-DFA} implies that the FM-index of a DFA (Theorem \ref{th:FM-index}) can be built in polynomial time while matching the co-lexicographic width $\bar p$ of the automaton.

\newpage

\appendix

\section{Proof of Lemma \ref{lem:string-nodes}}\label{app:string-nodes}

\paragraph{\textbf{Statement}}
Let $ \mathcal{A} = (Q, E, \Sigma, s, F, \le) $ be a CNFA. Let $ u, v \in Q $ and $ \alpha, \beta, \in Pref (\mathcal{L(A)}) $ such that $ u \in I_\alpha $, $ v \in I_\beta $ and $ \{u, v \} \not \subseteq I_\alpha \cap I_\beta $.
\begin{enumerate}
\item If $ \alpha < \beta $, then $ u ~\|~ v $ or $ u < v $.
\item If $ u < v $, then $ \alpha \prec \beta $.
\end{enumerate}

\paragraph{\textbf{Proof}}
Since $ \{u, v\} \subseteq I_\alpha \cap I_\beta $, then either $ u \in I_\alpha \setminus I_\beta $ or $ v \in I_\beta \setminus I_\alpha $. Hence $ \alpha \not = \beta $ and $ u \not = v $.
\begin{enumerate}
\item We proceed by induction on $ \min( |\alpha|, | \beta|) $. If $ \min( |\alpha|, | \beta|) = 0 $, then $ \alpha = \epsilon $, so $ u = s $. We conclude $ u = s < v $ by Axiom 1.

Now assume $ \min( |\alpha|, | \beta|) \ge 1 $. This implies $ \alpha \not = \epsilon \not = \beta $ and $ u \not = s \not = v $. Let $ a $ be the last letter of $ \alpha $ and let $ b $ the last letter of $ \beta $; it must be $ a \le b $. If $ a < b $, then $ \lambda(u) < \lambda(v) $, which implies $ u < v $ by Axiom 1. Otherwise, we can write $ \alpha = \alpha' e $ and $ \beta = \beta' e $, with $ e \in \Sigma $, $ \alpha', \beta' \in \Sigma^* $ and $ \alpha' < \beta' $. Let $ u', v' \in  Q $ such that $ u' \in I_{\alpha'} $, $ v' \in I_{\beta'} $, $ u \in \delta (u', e) $, $ v \in \delta(v', e) $. Then $ \{u', v' \} \not \subseteq I_{\alpha'} \cap I_{\beta'} $, otherwise $ \{u, v \} \subseteq I_\alpha \cap I_\beta $. By the inductive hypothesis, we have $ u' ~\|~ v' $ or $ u' < v' $. Hence it must be $ u ~\|~ v $ or $ u < v $, otherwise it would be $ v < u $, which implies $ v' \le u' $ by Axiom 2.

\item We know that $ \alpha \not = \beta $. If it were $ \beta < \alpha $, then by the previous part it would be $ v ~\|~ u $ or $ v < u $, leading to a contradiction. \qed
\end{enumerate}

\section{Proof of Corollary \ref{cor:1}}\label{app:cor1}

\paragraph{\textbf{Statement}}
Let $ \mathcal{A} = (Q, E, \Sigma, s, F, \le) $ be a CNFA. Let $ \alpha \in \Sigma^*$. Then, $ I_\alpha $ is a $ \le $-interval.

\paragraph{\textbf{Proof}}
We can prove this corollary also using Lemma \ref{lem:string-nodes}. Assume that $ u, v, z \in Q $ are such that $ u < v < z $ and $ u, z \in I_\alpha $. Suppose by contradiction that $ v \not \in I_\alpha $. Let $ \beta \in \Sigma^* $ be such that $ v \in I_\beta $. Since $ v \in I_\beta \setminus I_\alpha $, then $ \{u, v \} \not \subseteq I_\alpha \cap I_\beta $. By Lemma \ref{lem:string-nodes}, we have $ \alpha < \beta $. Similarly one obtains $ \beta < \alpha $. Hence $ \alpha < \alpha $, a contradiction. \qed

\section{Proof of Lemma \ref{lem5}}

\paragraph{\textbf{Statement}}
Let $ \mathcal{A} = (Q, E, \Sigma, s, F, \le) $ be a CNFA. Let $ \alpha \in \Sigma^*$, and let $ U $ be a $ \le $-interval of states. Then, the set $ U' $ of all states in $ Q $ that can be reached from $ U $ by following edges whose labels, when concatenated, yield $ \alpha $, is still a $ \le $-interval.

\paragraph{\textbf{Proof}}
We proceed by induction on $ | \alpha | $. If $ |\alpha| = 0 $, then $ \alpha = \epsilon $ and we are done. Now assume $ |\alpha| \ge 1 $. We can write $ \alpha = \alpha' a $, with $ \alpha' \in \Sigma^* $, $ a \in \Sigma $. Let $ u, v, z \in Q $ such that $ u < v < z $ and $ u, z \in U' $. We must prove that $ v \in U' $. By the inductive hypothesis, the set $ U'' $ of all states in $ Q $ that can be reached from some state in $ U $ by following edges whose labels, when concatenated, yield $ \alpha' $, is a $ \le $-interval. In particular, there exist $ u', z' \in U'' $ such that $ (u', u), (z', z) \in E $. Since $ u, z \in U' $, then $ \lambda (u) = a = \lambda (z) $. This implies that $ \lambda (v) = a $, because if for example it were $ \lambda (v) < a $, then we would have $ \lambda (v) < \lambda (u) $, which by Axiom 1 would imply $ v < u $, a contradiction. Since $ \lambda (v) = a $, then $ v $ must have at least one incoming edge $ (v', v) \in E $. By Axiom 2, we have $ u' \le v' \le z' $. Since $ u', z' \in U'' $ and $ U'' $ is a $ \le $-interval, then $ v' \in U''$, and so $ v \in U' $. \qed

\section{Proof of Lemma \ref{lem:lower bound}}\label{app:lower bound}

\paragraph{\textbf{Statement}}
The following hold: (1) Any partial order $(V,\leq)$ of width $ p $  has at least $2^p$ distinct $\leq$-intervals.
(2) For any size $n=|V|$ and any $1\leq p \leq n$, there exists a partial order $(V,\leq)$ of width $p$ with at least $(n/p)^p$ distinct $\leq$-intervals.
\paragraph{\textbf{Proof}}
Recall that the width $p$ of $(V,\leq)$ is defined as the size of its largest antichain $A$. (1) It is easy to see that any subset $I \subseteq A$ is a distinct $\leq$-interval. The bound $2^p$ follows. 
(2) Consider a partial order formed by $p$ chains $V_i = \{u_1^i < \dots < u_{n_i}^i\}$, for $i=1,\dots, p$ such that $u_a^i\ \|\ u_b^j$ for any $i\neq j$, $1\leq a \leq n_i$, and $1\leq b \leq n_j$. Then, any combination of $\leq_{V_i}$-intervals forms a distinct $\leq$-interval. On the $i$-th chain, there are $(n_i+1)n_i/2 + 1$ distinct intervals.
It follows that the number of distinct $\leq$-intervals is $\prod_{i=1}^p ((n_i+1)n_i/2+1) \geq \prod_{i=1}^p n_i $ (this simplification is motivated by Remark \ref{rem:lower bounds bit}, where we will take the logarithm of this quantity). By AM-GM inequality this quantity is maximized when all $n_i$ are equal to $n/p$, yielding at least $(n/p)^p$ distinct intervals. It is immediate to see that on such an order the size of the largest antichain is $p$ (simply take one element per chain). \qed

\section{Proof of Theorem \ref{th:store NFA}}\label{app:store NFA}

\paragraph{\textbf{Statement}}
The BWT of a $p$-sortable NFA $ \mathcal{A} = (Q, E, \Sigma, s, F) $ can be stored as an invertible representation using $|E|(\lceil \log \sigma \rceil + 2\lceil\log p\rceil + 2) + |Q|$ bits.

\paragraph{\textbf{Proof}} Using two bitvectors of length $|E|$, we mark with a bit set the last element in every list $\mathtt{OUT}[i]$ and $\mathtt{IN}[i]$ in order to encode their lengths. The remaining components take trivially $\lceil \log \sigma \rceil + 2\lceil\log p\rceil$ bits per transition. Bitvector $\tt FINAL$ takes $|Q|$ bits. 

We now show how to invert the representation.
First, assign the numbering $v_i = i$ to the states.
Using $\tt OUT$ and $\tt IN$, we can reconstruct functions $\pi$ (chain number of each state) and $\lambda$ (incoming labels of each state): scan $\tt OUT$ and count how many edges enter each chain; then, combine this information with the in-degrees of the nodes (sequence $\tt IN$) to reconstruct function $\pi$. This is possible since incoming edges in $\tt IN$ are sorted by increasing chain. Similarly, once reconstructed $\pi$ scan $\tt OUT$ and collect the number of edges labeled $c$, for every $c \in \Sigma$, that enter each chain, and combine this information with the in-degrees of the nodes to reconstruct function $\lambda$. This is possible since incoming edges in $\tt IN$ that enter the same chain are sorted by increasing letter.
This yields the block partition of the adjacency matrix shown in Figure \ref{tab:adj} using shades of gray. At this point, we use the monotonicity property of each block. 
Let $\mathtt{IN_{k,c}}$ denote the sub-sequence of $\tt IN$ corresponding to states $u$ with $\pi(u)=k$ and $\lambda(u)=c$.
For $i=1, \dots, n$, extract the pairs from $\mathtt{OUT}[i]$. For each such pair $(k,c)$, extract the leftmost element equal to $\pi(v_i)$ from $\mathtt{IN_{k,c}}$, and let $j$ be its column number. Finally, insert an edge labeled $c$ at coordinate $(i,j)$.\qed

\section{Proof of Theorem \ref{th:FM-index}}\label{app:FM-index}

\paragraph{\textbf{Statement}}
The BWT of a $p$-sortable NFA $ \mathcal{A} = (Q, E, \Sigma, s, F) $ can be encoded with a data structure of $|E|(\lceil \log\sigma\rceil + \lceil\log p\rceil + 2)\cdot (1+o(1)) + 2|Q|\cdot (1+o(1))$ bits that, given a query string $\alpha\in \Sigma^m$,  supports the following operations in $O(m\cdot p^2 \cdot \log(p\cdot \sigma))$ time:  
\begin{itemize}
    \item[(i)] Count the number of states reached by a path labeled $\alpha$.
    \item[(ii)] Return the states reached by a path labeled $\alpha$ as $p$ ranges on the $p$ chains.
    \item[(iii)] Decide whether $\alpha \in \mathcal L(\mathcal A)$.
\end{itemize}
Given any $I_\alpha$ and $c\in \Sigma$, the structure can furthermore compute $I_{\alpha\cdot c}$ in $O(p^2 \cdot \log(p\cdot \sigma)$ time.

\paragraph{\textbf{Proof}}
The index relies on the path coherency property of Lemma \ref{lem:path coherence - chains}. Intuitively, consider the ranges $U_1, \dots, U_p$ of states reached by string $\alpha$ on the $p$ chains $ Q_1, \dots, Q_p$. Given a character $c$, the goal is to update those ranges to the sets  $U'_1, \dots, U'_p$ of states reached by string $\alpha\cdot c$. By Lemma \ref{lem:path coherence - chains}, $U'_1, \dots, U'_p$ are indeed ranges on the $p$ chains. This extension step is a straightforward generalization of the search mechanism used on classic FM-indexes \cite{ferragina2000opportunistic} and, more in  general, on Wheeler graphs \cite{GAGIE201767}: forward search. 
Let $\mathtt{OUT}_i$ denote the subsequence of $\tt OUT$ corresponding to the $i$-th chain.
The forward search algorithm starts with $\alpha = \epsilon$ (the empty string), and updates the $\leq$-interval of the current string character by character by right-extensions.
We first describe the extension algorithm, and then discuss the data-structure details.
Let $\mathtt{OUT}_i[l_i,r_i]$ denote the range of states on the $i$-th chain reached by string $\alpha$ (that is, states in $U_i$). At the beginning ($\alpha = \epsilon$), each $\mathtt{OUT}_i[l_i,r_i] = \mathtt{OUT}_i[1,|\mathtt{OUT}_i|]$ is the full range. 
Assume we want to update those ranges by appending character $c$ to the current string $\alpha$.
For each $1\leq i,j \leq p$ we count the number of pairs of the form $(j,c)$ contained in subsequences $\mathtt{OUT}_i[1,l_i-1]$ and  $\mathtt{OUT}_i[1,r_i]$. Let $s_{i,j,c}$ and $e_{i,j,c}$ ("start" and "end", respectively) be such counters. For each chain $i$, we moreover count the number $C_{i,c}$ of edges $(u,v,a)$ entering the $i$-th chain (that is, $\pi(v)=i$) such that $a<c$ (this counter is known with the name "C array" in classic FM-indexes). Finally, consider the ordering $\prec_i$ of the edges entering the $i$-th chain
 defined as $(u,v) \prec_i (u',v')$ iff $v \leq v'$, with $\pi(v)=\pi(v')=i$ (edges reaching the same node can be ordered arbitrarily). Let $N_i(k)$ denote the integer such that the $k$-th edge of this ordering enters the $N_i(k)$-th state in the $i$-th chain. At this point, the forward search algorithm of \cite{ferragina2000opportunistic,GAGIE201767} generalizes as follows: for all $j=1, \dots, p$ the new interval $\mathtt{OUT}_j[l'_j,r'_j]$ of $\alpha\cdot c$ in the $j$-th chain is given by $l'_j = N_j( 1 + C_{j,c} + \sum_{i=1}^p s_{i,j,c})$ and $r'_j = N_j( C_{j,c} + \sum_{i=1}^p e_{i,j,c})$. Note that such an interval $\mathtt{OUT}_j[l'_j,r'_j]$ could be empty (more precisely, $r'_j = l'_j-1$) if there are no states in the $j$-th chain that are reached by $\alpha\cdot c$. If all intervals $\mathtt{OUT}_j[l'_j,r'_j]$, for $j=1\dots, p$ are empty, then no state of $\mathcal A$ is reached by $\alpha\cdot c$ and the search can stop. Otherwise, note that empty intervals have an important role: they indicate the position in the $j$-th chain where a state reached by $\alpha\cdot c$ would be placed, if it existed. This position is important as it is required to update the ranges on the other chains after a character extension. Crucially, note that we do not use the integers contained in $\tt IN$ (which are, in fact, not stored): we only use information regarding the length of each list $\mathtt{IN}[i]$ (required to compute function $N_i(k)$).

We now discuss the data-structure details for implementing efficiently the above procedure. 
To answer query (iii), we encode the bitvector $\tt FINAL$ (marking final states) with a representation supporting constant-time rank operations \cite{raman2007succinct}. 
Interval $I_\alpha$ can be obtained by running the forward search algorithm starting from the interval $[1,1]$ containing only the start state in the first chain, and the empty interval $[1,0]$ on the other chains. 
Once obtained the interval $I_\alpha$ for the query string $\alpha$, query (iii) can be supported in $O(p)$ time by simply checking, via $\tt FINAL$, if $I_\alpha$ contains final states.
We now discuss the structure needed to perform forward search.
Using two bitvectors supporting constant-time rank and select operations \cite{raman2007succinct}, we can mark the boundaries between (i) each $\mathtt{OUT}_i$ and, inside those subsequences, (ii) between each list $\mathtt{OUT}_i[k]$. Similarly, two bitvectors delimit the boundaries between (iii) each $\mathtt{IN}_i$ and, inside those subsequences, (iv) between each list $\mathtt{IN}_i[k]$. 
Bitvectors (ii) and (iv) have length $|E|$ bits. 
Bitvectors (i) and (iii) have length $|Q|$ bits and are actually equal (they both encode the lengths of the chains), so we need to store just one of them (even though for clarity in the following we treat them separately).  
As observed above, we do not actually store the content of $\tt IN$: the only information we keep about this list is contained in bitvectors (iii) and (iv).
Bitvectors (i) and (ii) provide constant-time random access to any element $\mathtt{OUT}_i[k]$.
Bitvectors (iii) and (iv) are used to implement function $N_i(k)$ in constant time with one select operation (on bitvector (iii)) and one rank operation (on bitvector (iv)). The pairs of $\tt OUT$ are treated as meta-characters on the extended alphabet $[1,p]\times \Sigma$, and are encoded in binary using $\lceil \log p \rceil + \lceil \log\sigma\rceil$ bits. We concatenate these pairs in a sequence $\mathtt{OUT'}$ of length $|E|$ and build a wavelet tree \cite{grossi2003high} on $\mathtt{OUT'}$ using this binary encoding. 
Wavelet trees support counting the number of pairs of the form $(j,c)$  in any prefix $\mathtt{OUT'}_i[1,t]$ with a simple rank operation taking $O(\log p + \log \sigma) = O(\log (p\cdot \sigma))$ time. Finally, we need to show how to compute $C_{j,c}$ for any $1\leq j \leq p$ and $c\in\Sigma$. This reduces to a \emph{range counting query} on the wavelet tree. Each element $(j,c) = \mathtt{OUT'}[k]$ can be considered as a two-dimensional point $(k,(j,c))$ (again, treat $(j,c)$ as an integer obtained by concatenating the binary representations of $j$ and $c$ of length $\lceil \log p \rceil$ and $\lceil \log\sigma\rceil$ bits, respectively). 
Let $\mathtt{OUT'}[p_i,q_i]$ be the range corresponding to pairs from $\mathtt{OUT}_i$, and note that this range can be retrieved in constant time using bitvectors (i) and (ii). Let moreover $z_{i,j,c}$ denote
the number of 
pairs $(j,a)$ contained in $\mathtt{OUT}_i$ such that $a<c$, i.e. the number of points contained in the two-dimensional range $[p_i,q_i] \times [(j,a),(j,a')]$, where $a$ is the smallest element of $\Sigma$ and $a'$ is the largest element of $\Sigma$ such that $a'<c$. The wavelet tree supports range counting queries in $O(\log (p\cdot \sigma))$ time \cite{NAVARRO20142}. Then, $C_{j,c} = \sum_{i=1}^p z_{i,j,c}$ can be computed in $O(p\cdot \log (p\cdot \sigma))$ time. Our claimed query time follows. As far as the space usage of our index is concerned, all data structures that we used are \emph{succinct}, i.e. they only use a low-order number of bits on top of the information-theoretic minimum required to store the underlying data. Our thesis follows. \qed

\section{Proof of Theorem \ref{th:BWT of DFA}}\label{app:BWT of DFA}

\paragraph{\textbf{Statement}}
The BWT of a $p$-sortable DFA $ \mathcal{A} = (Q, E, \Sigma, s, F) $ can be stored as an invertible representation using $|E|(\lceil \log \sigma \rceil + \lceil\log p\rceil + 2) + |Q|$ bits.

\paragraph{\textbf{Proof}}
We store the same information of Theorem \ref{th:FM-index}, except bitvectors (i) and (iii) and the additional structures supporting constant-time rank and select on all the sequences (we do not need fast queries since we are just describing an encoding). The resulting encoding uses the claimed space. Note that the explicitly-stored in-degrees of the nodes and the information contained in $\tt OUT$ is sufficient to reconstruct bitvectors (i) and (iii): simply count how many edges exit the $i$-th chain, and use the in-degrees to reconstruct the chain decomposition. At this point it is sufficient to note that, being $\mathcal A$ deterministic, forward search yields a singleton $I_\alpha = \{u_\alpha\}$ for any string $\alpha \in Pref(\mathcal L(\mathcal A))$. This fact can be used to perform a visit of the underlying graph starting from the source, thus reconstructing the transition function. \qed

\section{Parameterized analysis of the powerset construction algorithm}\label{app:powerset}

First, we need to introduce some notation. Corollary \ref{cor:path-coherence} motivates the following definition.

\begin{definition}
Let $ \mathcal{A} = (Q, E, \Sigma, s, F) $ be a $ p $-sortable NFA. Fix any $ \le $-chain decomposition $ \{Q_i \}_{i = 1}^p $ of $ Q $. For $ \alpha \in Pref (\mathcal{L(\mathcal{A})}) $ and for $ i = 1, \dots, p $, let $ I_\alpha^i $ be the (possibly empty) $ \le_{Q_i} $-interval being the intersection between $ I_\alpha $ and $ Q_i $. Moreover, for every $ i = 1, \dots, p $ define:
\begin{equation*}
Pref (\mathcal{L(\mathcal{A})})^i = \{\alpha \in Pref (\mathcal{L(\mathcal{A}))}\ |\ I_\alpha^i \not = \emptyset \}
\end{equation*}
and:
\begin{equation*}
I_{Pref (\mathcal{L(\mathcal{A})})}^i = \{I_\alpha^i\ |\ \alpha \in Pref (\mathcal{L(\mathcal{A})})^i \}.
\end{equation*}
\end{definition}

\begin{remark}\label{rem2}
Notice that for every $ \alpha \in Pref (\mathcal{L(\mathcal{A})}) $ there exists at least one $ i $ such that $ I_\alpha^i \not = \emptyset $.
\end{remark}

We will also need some definitions from \cite{alanko20regular}:

\begin{definition}
Let $ (V, \le ) $ be a total order.
\begin{enumerate}
    \item Let $ I, J $ be $ \le $-intervals, and assume that $ I \subseteq J $.
    \begin{enumerate}
        \item We say that $ I $ is a prefix of $ J $ if $ (\forall x \in I)(\forall y \in J \setminus I)(x < y) $;
        \item We say that $ I $ is a suffix of $ J $ if $ (\forall x \in I)(\forall y \in J \setminus I)(y < x) $;
    \end{enumerate}
    \item A family $ \mathcal{C} $ of nonempty $ \le $-intervals is said to be a \emph{prefix/suffix familiy} if for all $ I, J \in \mathcal{C} $ such that $ I \subseteq J $ we have that $ I $ is either a prefix or a suffix of $ J $.
\end{enumerate}
\end{definition}

In particular, we will use the following result \cite{alanko20regular}.

\begin{lemma}\label{lem:4}
Let $ (V, \le) $ be a finite total order, and let $ \mathcal{C} $ be a prefix/suffix family of nonempty $ \le $-intervals in $ V $.
\begin{enumerate}
    \item $ |\mathcal{C}| \le 2 |V| - 1 $;
    \item If for every $ I, J \in \mathcal{C} $ we define:
    \begin{equation*}
        I <^* J \text{ if and only if } (\exists x \in I)(\forall y \in J)(x < y) \lor (\exists y \in J)(\forall x \in I)(x < y)
    \end{equation*}
    then $ (\mathcal{C}, \le^*) $ is a total order.
\end{enumerate}
\end{lemma}

We are interested in prefix/suffix families because the following lemma shows that $I_{Pref (\mathcal{L(\mathcal{A})})}^i $ is a prefix/suffix family for every $ i $.

\begin{lemma}\label{lem:prefix-suffix family}
Let $ \mathcal{A} = (Q, E, \Sigma, s, F) $ be a $ p $-sortable NFA, and let $ \{Q_i \}_{i = 1}^p $ be a $ \le $-chain decomposition of $ Q $. Then, $ I_{Pref (\mathcal{L(\mathcal{A})})}^i $ is a prefix/suffix family of $ \le $-intervals in $ (Q_i, \le) $, for every $ i = 1, \dots, p $.
\end{lemma}

\begin{proof}
By definition $ I_{Pref (\mathcal{L(\mathcal{A})})}^i $ is a family of non-empty $ \le $-intervals in $ (Q_i, \le) $. Now assume that $ I^i_\alpha \subseteq I_\beta^i $ and suppose by contradiction that $ I^i_\alpha $ is neither a prefix nor a suffix of $ I_\beta^i $. Since $ I^i_\alpha $ is not a prefix of $ I_\beta^i $, then there exists $ u, v, \in Q_i $ such that $ u \in I^i_\alpha $, $ v \in I^i_\beta \setminus I_\alpha^i $ and $ v < u $. In particular $ u \in I_\alpha $, $ v \in I_\beta $ and $ \{u, v \} \not \subseteq I_\alpha \cap I_\beta $, so by Lemma \ref{lem:string-nodes} we conclude $ \beta < \alpha $. Similarly, using that $ I^i_\alpha $ is not a suffix of $ I_\beta^i $, we conclude $ \alpha < \beta $, a contradiction. \qed
\end{proof}

Lemma \ref{lem:prefix-suffix family} allows us to introduce the following definition.

\begin{definition}\label{definition1}
Let $ \mathcal{A} = (Q, E, \Sigma, s, F) $ be a $ p $-sortable NFA, and let $ \{Q_i \}_{i = 1}^p $ be a $ \le $-chain decomposition of $ Q $. We denote by $ \le^i $ the total order on $ I_{Pref (\mathcal{L(\mathcal{A})})}^i $ built in Lemma \ref{lem:4}.
\end{definition}

Similarly to Lemma \ref{lem:string-nodes}, we can now prove that $ <^i $ respects the co-lexicographic order on $ \Sigma^* $.

\begin{lemma}\label{lem:1}
Let $ \mathcal{A} = (Q, E, \Sigma, s, F) $ be a $ p $-sortable CNFA, and let $ \{Q_i\}_{i = 1}^p $ be a $ \le $-chain decomposition of $ Q $. For some $ i $, let $ I_\alpha^i, I_\beta^i \in I_{Pref (\mathcal{L(\mathcal{A})})}^i $.
\begin{enumerate}
    \item If $ I_\alpha^i <^i I_\beta^i $, then $ \alpha < \beta $.
    \item If $ \alpha < \beta $, then  $ I_\alpha^i \le^i I_\beta^i $.
\end{enumerate}
\end{lemma}

\begin{proof}
\begin{enumerate}
    \item $ I_\alpha <^i I_\beta $ implies that at least one of the following is true:
    \begin{enumerate}
        \item $ (\exists u \in I_\alpha^i)(\forall v \in I_\beta^i)(u < v) $;
        \item $ (\exists v \in I_\beta^i)(\forall u \in I_\alpha^i)(u < v) $.
    \end{enumerate}
This implies that there exist $ u \in I_\alpha^i $, $ v \in I_\beta^i $ such that $ u < v $ and either $ u \not \in I_\beta^i $ or $ v \not \in I_\alpha^i $. Hence $ u \in I_\alpha $, $ v \in I_\beta $ and $ \{u, v \} \not \subseteq I_\alpha \cap I_\beta $, so by Lemma \ref{lem:string-nodes}, we conclude $ \alpha \prec \beta $.
\item If it were $ I_\beta^i <^i I_\alpha^i $, then by the previous part it would be $ \beta \prec \alpha $. \qed
\end{enumerate}
\end{proof}

Let $ \mathcal{A} = (Q, E, \Sigma, s, F) $ be a $ p $-sortable NFA and let $ \mathcal{A^*} = (Q^*, E^*, \Sigma, s^*, F^*) $ be the DFA obtained from $ \mathcal{A} $ by powerset construction. We have already observed that the elements of $ Q^* $ must be the union of at most $ p $ intervals. Lemma \ref{lem:1} provides a further restriction to the elements of $ Q^* $, because if $ \alpha < \beta $, then  $ I_\alpha^i \le^i I_\beta^i $ must hold \textit{for every $ i $}. We now have all the elements required to prove our parameterization of the number of states of $ Q^* $.

\paragraph{\textbf{Statement of Theorem \ref{th:powerset}}}
Let $ \mathcal{A} = (Q, E, \Sigma, s, F) $ be a $ p $-sortable NFA, and let $ \mathcal{A^*} = (Q^*, E^*, \Sigma, s^*, F^*) $ be the DFA obtained from $ \mathcal{A} $ by the powerset construction algorithm. Then, $ |Q^*| \le 2^p (|Q| - p + 1) - 1 $.

\paragraph{\textbf{Proof}}
Let $ \{Q_i \}_{i = 1}^p $ be a $ \le $-chain decomposition of $ Q $. For every $ \alpha \in I_{Pref (\mathcal{L(\mathcal{A})})} $, define the $ p $-tuple:
\begin{equation*}
T_\alpha := (I_\alpha^1, \dots, I_\alpha^p)
\end{equation*}
where $ I_\alpha^i \not = \emptyset $ for at least one $ i $ (see Remark \ref{rem2}). If we define:
\begin{equation*}
    T = \{T_\alpha\ |\ \alpha \in I_{Pref (\mathcal{L(\mathcal{A})})} \}
\end{equation*}
then $ |Q^*| = |T | $.

For every nonempty $ K \subseteq \{1, \dots, p \} $, let $ T_K $ be the set of all $ T_\alpha $'s such that $ \alpha \in I_{Pref (\mathcal{L(\mathcal{A})})} $ and $ I_\alpha^i \not = \emptyset $ if and only if $ i \in K $. Clearly, $ T $ is the disjoint union of all $ T_K $'s, that is:
\begin{equation*}
    T = \bigsqcup_{\emptyset \subsetneqq K \subseteq \{1, \dots, p \}} T_K.
\end{equation*}

In the following, if $ I_\alpha^i \not = \emptyset $ (or equivalently, $ I_\alpha^i \in I_{Pref (\mathcal{L(\mathcal{A})})}^i) $, we identify $ I_\alpha^i $ with its position in the total order $ \le^i $ (see definition \ref{definition1}), hence $ I_\alpha^i $ will be an integer between $ 1 $ and $ |I_{Pref (\mathcal{L(\mathcal{A})})}^i| $.

Fix $ \emptyset \subsetneqq K \subseteq \{1, \dots, p \} $. Pick $ T_\alpha, T_\beta \in T_K $. Notice that there are no $ i, j \in K $ such that $ I_\alpha^i < I_\beta^i $ and $ I_\alpha^j > I_\beta^j $, because Lemma \ref{lem:1} would imply $ \alpha < \beta $ and $ \beta < \alpha $, a contradiction. This means that for every integer $ s $ there exists at most one element $ T_\alpha $ in $ T_K $ such that $ \sum_{i \in K} I_\alpha^i = s $, because if there existed $ T_\alpha, T_\beta \in T_K $, with $ T_\alpha \not = T_\beta $, such that $ \sum_{i \in K} I_\alpha^i = \sum_{i \in K} I_\beta^i $, then there would exist $ i, j \in K $ such that $ I_\alpha^i < I_\beta^i $ and $ I_\alpha^j > I_\beta^j $. Since for every $ T_\alpha \in T_K $ it holds $ |K| \le \sum_{i \in K} I_\alpha^i \le \sum_{i \in K} |I_{Pref (\mathcal{L(\mathcal{A})})}^i| $, we conclude:
\begin{equation*}
|T_K| \le (\sum_{i \in K} |I_{Pref (\mathcal{L(\mathcal{A})})}^i| ) - |K| + 1.
\end{equation*}
We can then write:
\begin{equation*}
|Q^*| = |T| = \sum_{\emptyset \subsetneqq K \subseteq \{1, \dots, p \}} |T_K| \le \sum_{\emptyset \subsetneqq K \subseteq \{1, \dots, p \}} ((\sum_{i \in K} |I_{Pref (\mathcal{L(\mathcal{A})})}^i| ) - |K| + 1)
\end{equation*}
Notice that $ \sum_{\emptyset \subsetneqq K \subseteq \{1, \dots, p \}} \sum_{i \in K} |I_{Pref (\mathcal{L(\mathcal{A})})}^i| = 2^{p - 1} \sum_{i = 1}^p |I_{Pref (\mathcal{L(\mathcal{A})})}^i| $ because every $ i \in \{1, \dots, p \} $ occurs in exactly $ 2^{p - 1} $ subsets of $ \{1, \dots, p \} $. Similarly, we have $ \sum_{\emptyset \subsetneqq K \subseteq \{1, \dots, p \}} |K| = 2^{p - 1} p $. Hence:
\begin{equation*}
    |Q^*| \le (2^{p - 1} \sum_{i = 1}^p |I_{Pref (\mathcal{L(\mathcal{A})})}^i|) - p 2^{p - 1} + 2^{p} - 1.
\end{equation*}
By Lemma \ref{lem:prefix-suffix family} we know that $ I_{Pref (\mathcal{L(\mathcal{A})})}^i $ is a prefix/suffix family of $ \le $-intervals in $ (Q_i, \le) $, so by Lemma \ref{lem:4} we obtain that $ |I_{Pref (\mathcal{L(\mathcal{A})})}^i| \le 2 |Q_i| - 1 $. We conclude:
\begin{equation*}
    |Q^*| \le (2^{p - 1} \sum_{i = 1}^p (2 |Q_i| - 1)) - p 2^{p - 1} + 2^{p} - 1 = 2^p |Q| - p2^{p - 1} - p 2^{p - 1} + 2^{p} - 1 = 2^p (|Q| - p + 1) - 1.
\end{equation*} 
which proves our claim. \qed

\section{DFAs admit a unique maximal co-lexicographic order}\label{app:DFA unique}

We aim to prove that a DFA admits a unique maximal co-lexicographic order. First, we prove the following lemma, which can be used to identify a co-lexicographic order.

\begin{lemma}\label{lem:2}
Let $ \mathcal{A} = (Q, E, \Sigma, s, F) $ be an NFA, and let $ \le $ be a reflexive and antisymmetric relation on $ V $ that satisfies the following properties:
\begin{enumerate}
    \item For every $ u, v \in Q $, if $\lambda(u) < \lambda(v) $, then $ u < v $;
    \item For all edges $ (u', u), (v', v) \in E $ such that $ \lambda (u) = \lambda (v) $, if $ u < v $, then $ u' \leq v' $.
\end{enumerate}
Let $ \le^* $ be the transitive closure of $ \le $, and assume that $ \le^* $ is antisymmetric. Then, $ \le^* $ is a co-lexicographic order of $ \mathcal{A} $.
\end{lemma}

\begin{proof}
First, $ \le^* $ is reflexive because $ \le $ is reflexive, so $ \le^* $ is a partial order. Morover, $ \le^* $ satisfies Axiom 1, because if $ u, v \in Q $ are such that $ \lambda (u) < \lambda (v) $, then $ u < v $ and so $ u <^* v $. So we just have to prove that Axiom 2 is satisfied. Consider two edges $ (u', u), (v', v) \in E $ such that $ \lambda (u) = \lambda (v) $ and $ u <^* v $; we must prove that $ u' \le^* v' $. Since $ \le^* $ is the transitive closure of $ \le $, there exist states $ z_1, \dots, z_r $ ($ r \ge 0 $) such that $ u < z_1 $, $ z_1 < z_ 2 $, $ \dots $, $ z_r < v $, and in particular $ u <^* z_1 <* z_2 <^* \dots <^* z_ r <^* v $. Since $ \lambda (u) = \lambda (v) $, then $ \lambda (u) = \lambda (z_1) = \dots = \lambda (z_r) = \lambda (v) $. Indeed, if for some $ j $ it were for example $ \lambda (z_j) > \lambda (u) = \lambda (v) $, then it should be $ v < z_j $ and so $ v <^* z_j $, which contradicts $ z_j <^* v $. In particular, since $ u $ and $ v $ have ingoing edges, then even all $ z_i$'s have ingoing edges $ (z'_i, z_i) \in E $, for $ i = 1, \dots, k $. The second assumption implies that $ u' \le z'_1 $, $ z'_1 \le z'_2 $, $ \dots $, $ z'_k \le v' $, so $ u' \le^* z'_1 \le^* z'_2 \dots z'_k \le^* v' $ and we conclude $ u' \le^* z' $. \qed
\end{proof}

The task of determining a co-lexicographic order of smallest width can be simplified if every co-lexicographic order must be the restriction of some total order on the set of states. This motivates the following definition.

\begin{definition}
Let $ \mathcal{A} = (Q, E, \Sigma, s, F) $ be an NFA. We say that a total order $ \le_\# $ on Q is an \emph{underlying order} of $ \mathcal{A} $ if for every co-lexicographic order $ \le $ of $ \mathcal{A} $:
\begin{equation*}
    u < v \implies u <_\# v \quad \text{ $ \forall  u, v \in Q $}.
\end{equation*}
\end{definition}


In general an NFA does not admit an underlying order: simply consider a source $ s $ connected by the same label to $ n $ non-adjacent states. However, for DFAs we have the following result:

\begin{lemma}\label{lem:underlying order DFA}
Let $ \mathcal{A} = (Q, E, \Sigma, s, F) $ be a DFA. Then, $ \mathcal{A} $ admits an underlying order $ <_\# $. In particular, for every $ u, v \in Q $, if $ \lambda (u) < \lambda (v) $, then $ u <_\# v $.
\end{lemma}

\begin{proof}
Let $ Q = \{u_1, \dots, u_n \} $. For every $ i = 1, \dots, n $ let $ \alpha_i \in Pref (\mathcal{L(\mathcal{A})}) $ be such that $ q_i \in I_{\alpha_i} $. Intuitively, the $ \alpha_i $'s can be determined by building a directed spanning tree of $ \mathcal{A} $ with root $ s $. Since $ \mathcal{A} $ is a DFA, then $ \alpha_1, \dots, \alpha_n $ are pairwise distinct. Without loss of generality, assume $ \alpha_1 < \alpha_2 < \dots < \alpha_n $. Lemma \ref{lem:string-nodes} implies that, for every co-lexicographic order $ \le $ of $ \mathcal{A} $, if $ u_i $ and $ u_j $, with $ i < j $, are $ \le $-comparable, then it must be $ u_i < u_j $. As a consequence, if $ \le_\# $ is the total order on $ Q $ such that $ u_1 <_\# u_2 <_\# \dots <_\# u_n $, then $ \le_\# $ is an underlying order of $ \mathcal{A } $. Since every automaton admits a co-lexicographic order, the final statement follows from Axiom 1. \qed
\end{proof}

\begin{remark}
In general, $ \le_\# $ is not unique, see figure \ref{fig:total order}. Nonetheless, we have uniqueness in the following sense: if there exists a co-lexicographic order $ \le $ for which $ u $ and $ v $ are $ \le $-comparable, then the mutual order of $ u $ and $ v $ with respect to $ \le_\# $ is uniquely determined. This is consistent with figure \ref{fig:total order}: for every co-lexicographic order $ \le $, states $ q_1 $ and $ q_2 $ cannot be $ \le $-comparable by Axiom 2.
\end{remark}

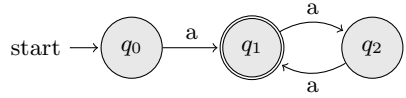
\begin{figure}[h!]
\centering
	\centering
	\begin{tikzpicture}[shorten >=1pt,node distance=1.6cm,on grid,auto]
	\tikzstyle{every state}=[fill={rgb:black,1;white,10}]
	
	\node[state,initial]   (q_0)                    {$ q_0 $};
	\node[state,accepting] (q_1)  [right of=q_0]    {$ q_1 $};
	\node[state]           (q_2)  [right of=q_1]    {$ q_2 $};
	
	\path[->]
	(q_0) edge node {a}    (q_1)
	(q_1) edge [bend left] node {a}    (q_2)
	(q_2) edge [bend left] node {a}    (q_1);
	\end{tikzpicture}
\caption{Observe that $ q_1 \in I_{a} $ and $ q_2 \in I_{aa} $, so we can let $ q_1 <_\# q_2 $. On the other hand, $ q_1 \in I_{aaa} $ and $ q_2 \in I_{aa} $, so we can also let $ q_2 <_\# q_1 $.}\label{fig:total order}
\end{figure}


We can now prove that a DFA admits a unique maximal co-lexicographic order.

\begin{theorem}\label{theor2}
Let $ \mathcal{A} = (Q, E, \Sigma, s, F) $ be a DFA. Then, $ \mathcal{A} $ admits a unique maximal co-lexicographic order.
\end{theorem}

\begin{proof}
By Lemma \ref{lem:4}, $ \mathcal{A} $ admits an underlying order $ \le_\# $. Consider two maximal co-lexicographic orders $ \le_1 $, $ \le_2 $ of $ \mathcal{A} $; we want to prove that $ \le_1 $ and $ \le_2 $ are equal. Let $ \le_3 $ be the union of $ \le_1 $ and $ \le_2 $ (that is, $ u \le_3 v $ if and only if $ (u \le_1 v) \lor (u \le_2 v) $) and let $ \le_4 $ be the transitive closure of $ \le_3 $. Then $ \le_3 $ is reflexive (because e.g. $ \le_1 $ is reflexive); moreover, $ \le_3 $ and $ \le_4 $ are antisymmetric (because they are restrictions of the underlying order $ \le_\# $). Notice that $ \le_3 $ satisfies the hypotheses of Lemma \ref{lem:2}, because:
\begin{enumerate}
\item If $ \lambda (u) < \lambda (v) $, then e.g. $ u \le_1 v $ (since $ \le_1 $ is a co-lexicographic order) and so $ u \le_3 v $;
\item If $ (u', u), (v, v) \in E $ are such that $ \lambda (u) = \lambda (v) $ and $ u \le_3 v $, then we have $ (u \le_1 v) \lor (u \le_2 v) $, which implies $ (u' \le_1 v') \lor (u' \le_2 v') $ (since $ \le_1 $ and $ \le_2 $ are co-lexicographic orders) and so $ u' \le_3 v' $.
\end{enumerate}
By Lemma \ref{lem:2} we can then conclude that $ \le_4 $ is a co-lexicographic order. However, $ \le_4 $ is a refinement of both $ \le_1 $ and $ \le_2 $, which are maximal, so $ \le_4 $ must be equal to both $ \le_1 $ and $ \le_2 $. The conclusion follows. \qed
\end{proof}

Let us present a constructive proof of Theorem \ref{theor2}, which can be used to build the maximal co-lexicographic order of a DFA.

\begin{theorem}\label{theor3}
Let $ \mathcal{A} = (Q, E, \Sigma, s, F) $ be a DFA, and let $ \le_\# $ be an underlying order of $ \mathcal{A} $. Let $ \le $ be the reflexive restriction of $ \le_\# $ such that for all states $ u, v \in Q $ with $ u <_\# v $ it holds $ u ~\|~ v $ if and only if for some $ r \ge 1 $ there exist states $ u_0, u_1, \dots, u_r $ and $ v_0, v_1, \dots, v_r $ with the following properties:
\begin{enumerate}
\item $ u_r = u $ and $ v_r = v $;
\item $ (u_k, u_{k + 1}), (v_k, v_{k + 1}) \in E $ for $ k = 0, 1, \dots, r - 1 $;
\item $ \lambda (u_k) = \lambda (v_k) $ for $ k = 1, 2, \dots, r $ (in particular, $ \lambda (u) = \lambda (v) $);
\item $ u_k <_\# v_k $ for $ k = 1, 2, \dots, r $;
\item $ v_0 <_\# u_0 $.
\end{enumerate}
Then, $ \le $ is a maximal co-lexicographic order of $ \mathcal{A} $. Moreover, $ \le $ is the unique maximal co-lexicographic order of $ \mathcal{A} $.
\end{theorem}

\begin{proof}
In order to prove that $ \le $ is a co-lexicographic order, we must prove that $ \le $ is a partial order that satisfies Axioms 1 and 2. First, let us prove that $ \le $ satisfies Axioms 1 and 2.
\begin{enumerate}
\item Let $ u, v \in Q $ be such that $ \lambda (u) < \lambda (v) $. By Lemma \ref{lem:underlying order DFA}, we have $ u <_\# v $. The definition of $ \le $ implies that $ u < v $.
\item Consider edges $ (u', u), (v', v) \in E $ such that $ \lambda (u) = \lambda (v) $, $ u < v $ (and so $ u <_\# v $) and $ u' \not = v' $. It cannot hold $ v' < u' $ (and so $ v' <_\# u') $ otherwise it should be $ u ~\|~ v $. Similarly, it cannot be $ u'~\|~v' $, otherwise both $ u' <_\# v' $ and $ v' <_\# u' $ would imply $ u ~\|~ v $. We conclude that it it must be $ u' < v' $.
\end{enumerate}

Second, let us prove that $ \le $ is a partial order. Reflexivity follows by definition, and antisymmetry is immediate because $ \le $ is a restriction of the total order $ \le_\# $. Let us prove transitivity.

Assume $ u < v $ and $ v < z $ . In particular $ u <_\# v $ and $ v <_\# z $, and so $ u <_\# z $. It must be $ \lambda(u) \le \lambda (v) \le \lambda (z) $, because for example $ \lambda (v) < \lambda (u) $ would imply $ v < u $ by Axiom 1, which again contradicts antisymmetry. If $ \lambda (u) < \lambda (z) $, then $ u < z $ by Axiom 1 and we are done. Otherwise, we have $ \lambda (u) = \lambda (z) $. Assume by contradiction that it is not true that $ u < z $. Since $ \le $ is a restriction of $ \le_\# $, it should be $ u ~\|~ z $. This means that there exist states $ u_0, u_1, \dots, u_r $ and $ z_0, z_1, \dots, z_r $ such that $ u_r = u $, $ v_r = v $, $ (u_k, u_{k + 1}), (z_k, z_{k + 1}) \in E $ for $ k = 0, 1, \dots, r - 1 $, $ \lambda (u_k) = \lambda (z_k) $ for $ k = 1, 2, \dots, r $, $ u_k <_\# z_k $ for $ k = 1, 2, \dots, r $ and $ z_0 <_\# u_0 $.

Let us prove that for every $ k = 1, \dots, r $ there exists a state $ v_k $ such that $ u_k < v_k $ and $ v_k < z_k $. We proceed by induction on $ h := r - k $. If $ h = 0 $, then just pick $ v_r := v $. Now assume that for $ k \ge 2 $ there exists $ v_k $ such that $ u_k < v_k $ and $ v_k < z_k $. We want to prove that there exists a state $ v_{k + 1} $ such that $ u_{k - 1} < v_{k - 1} $ and $ v_{k - 1} < z_{k - 1} $. Since $ \lambda (u_k) = \lambda (z_k) $, as usual we conclude $ \lambda (u_k) = \lambda (v_k) = \lambda (z_k) $. In particular, since only $ s $ has not incoming edges and $ u_k $, $ v_k $, $ z_k $ are distinct, then there exists a state $ v_{k - 1} $ such that $ (v_{k - 1}, v_{k}) \in E $. Since $ u_k < v_k $, $ v_k < z_k $ and $ \lambda (u_k) = \lambda (v_k) = \lambda (z_k) $, Axiom 2 implies $ u_{k - 1} \le v_{k - 1} $ and $ v_{k - 1} \le z_{k - 1} $. Let us prove that $ u_{k - 1} < v_{k - 1} $ and $ v_{k - 1} < z_{r - 1} $. If it were $ u_{k - 1} = v_{k - 1} $, then $ v_{k - 1} ~\|~ z_{k - 1} $ (a contradiction) because clearly it holds $ u_{k - 1} ~\|~ z_{k - 1} $. Similarly, assuming $ v_{k - 1} = z_{k - 1} $ leads to a contradiction. The proof by induction is then complete.

In particular, we know that there exists a state $ v_1 $ such that $ u_1 < v_1 $ and $ v_1 < z_1 $, with $ \lambda (u_1) = \lambda (v_1) = \lambda (z_1) $. Once again, there exist a state $ v_0 $ such that $ (v_0, v_1) \in E $, and Axiom 2 implies $ u_0 \le v_0 $ and $ v_0 \le z_0 $. Hence $ u_0 \le_\# v_0 \le_\# z_0 $, which contradicts $ z_0 <_\# u_0 $. The proof of transitivity is then complete.

Finally, $ \le $ let us prove that $ \le $ is a maximal co-lexicographic order and it is the unique maximal co-lexicographic order. To this end, it will suffice to prove that if $ u ~\|~ v $, with $ u <_\# v $, then $ u $ and $ v $ are not $ \le_1 $-comparable \emph{for every co-lexicographic order $ \le_1 $} (note that if $ u $ and $ v $ were $ \le_1 $-comparable, it should be $ u <_1 v $ because $ \le_1 $ must be a restriction of $ \le_\# $). We know that there exist states $ u_0, u_1, \dots, u_r $ and $ v_0, v_1, \dots, v_r $ such that $ u_r = u $, $ v_r = v $, $ (u_k, u_{k + 1}), (v_k, v_{k + 1}) \in E $ for $ k = 0, 1, \dots, r - 1 $, $ \lambda (u_k) = \lambda (v_k) $ for $ k = 1, 2, \dots, r $, $ u_k <_\# v_k $ for $ k = 1, 2, \dots, r $ and $ v_0 <_\# u_0 $. We proceed by induction on $ r $. If $ r = 0 $, and if it were $ u <_1 v $, then Axiom 2 would imply $ u_0 \le_1 v_0 $ (because $ \lambda (u) = \lambda (v) $ and $ (u_0, u), (v_0, v) \in E $), so $ u_0 \le_\# v_0 $, which contradicts $ v_0 <_\# u_0 $. Now let $ r \ge 1 $. Clearly $ u_{r - 1} ~\|~ v_{r - 1} $ and by the inductive step $ u_{r - 1} $ and $ v_{r - 1} $ are not $ \le_1 $-comparable. By Remark \ref{rem1}, we conclude that $ u $ and $ v $ are not $ \le_1 $ comparable. \qed
\end{proof}

From Theorem \ref{theor3} we easily derive a polynomial algorithm to build the unique maximal co-lexicographic order of a DFA.

\paragraph{\textbf{Statement of Theorem \ref{th:complexity co-lex DFA}}}
Let $ \mathcal{A} = (Q, E, \Sigma, s, F) $ be a DFA. We can find the unique maximal co-lexicographic order of $ \mathcal{A} $ in $ O(|E|^2) $ time.
\paragraph{\textbf{Proof}}
By Lemma \ref{lem:underlying order DFA}, $ \mathcal{A} $ admits an underlying order. Following \cite[Thm. 4]{alanko20regular}, in $ O(|E|) $ time we can build an underlying order $ \le_\# $ by prefix-sorting a directed spanning tree of $ \mathcal{A}$ with source $ s $. Consider the graph $ G = (V, F) $, where $ V  = \{(u, v)\ |(\lambda (u) = \lambda (v)) \land  (u <_\# v) \} $ 
and $ F = \{((u',v'),(u,v)) \in V\times V \ |\ (u',u),(v',v)\in E \}$.
Intuitively, we will use $G$ to propagate the incomparability relation $\|$ between pairs of states of $\mathcal A$. First, 
for all pairs of edges $ (u', u), (v', v) \in E $ such that $ \lambda (u) = \lambda (v) $, $ u <_\# v $, and $ v' <_\# u' $, mark node $ (u, v) $ of $G$. This process takes $O(|E|^2)$ time. Finally, mark all nodes reachable on $ G $ from marked nodes. This can be done with a simple DFS visit of $G$, initiating the stack with all marked nodes. Also this process takes $O(|E|^2)$ time. By Theorem \ref{theor3}, if we remove from $ \le_\# $ the set of all marked pairs of $V$ we obtain the maximal co-lexicographic order of $ \mathcal{A} $. 
\qed

\bibliographystyle{plainurl}
\bibliography{regindex}

\end{document}